\documentclass[a4paper,12pt]{article}
\usepackage[utf8]{inputenc}
\usepackage[normalem]{ulem}
\usepackage{amssymb}
\usepackage{authblk}

\usepackage{graphicx}
\usepackage{graphics}
\usepackage{fancyhdr}
\usepackage{amsmath}
\usepackage{amsfonts}
\usepackage{amssymb}
\usepackage{amsthm}
\usepackage{epsfig}
\usepackage{lipsum} %
\usepackage{color}
\usepackage{amsthm}
\usepackage{sectsty} 
\usepackage{hyperref}
\usepackage{comment}
\usepackage{url}
\usepackage{empheq}
\usepackage{eurosym}
\usepackage{mathrsfs}
\usepackage[english]{babel}
\usepackage{booktabs}
\usepackage{caption}
\usepackage{subfig}
\usepackage{bm}
\usepackage{tikz} 
\usetikzlibrary{arrows,decorations.pathmorphing,backgrounds,fit,positioning,shapes.symbols,chains}

\usepackage{xcolor}
\usepackage{color}
\usepackage{enumitem}
\usetikzlibrary{shapes,arrows}
\usepackage{listings}
\usepackage{rotating}

\usepackage[margin=1in]{geometry}

\theoremstyle{definition} \newtheorem{de}{Definition}[section]
\theoremstyle{plain}      \newtheorem{te}[de]{Theorem}
\theoremstyle{remark}     \newtheorem{os}[de]{Remark}
\theoremstyle{plain}      \newtheorem{pr}[de]{Proposition}
\theoremstyle{plain}      \newtheorem{lem}[de]{Lemma}
\theoremstyle{plain}	  \newtheorem{co}[de]{Corollary}
\theoremstyle{definition} 
\theoremstyle{remark}		\newtheorem{es}[de]{Example}
\theoremstyle{remark}		
\theoremstyle{plain}        \newtheorem*{beware*}{Beware}

\newcommand{\numberset}{\mathbb}
\newcommand{\N}{\numberset{N}}
\newcommand{\R}{\numberset{R}}

\newcommand{\p}{\numberset{P}}
\newcommand{\E}{\numberset{E}}
\newcommand{\T}{\numberset{T}}

\newcommand{\F}{\mathscr{F}}

\newcommand{\Saff}{S_{t_k}^{\bm{\xi},\bm{\eta}}}

\usepackage[round,sort&compress,comma]{natbib}

\usepackage[titletoc]{appendix}
\usepackage{titlesec}
\titlelabel{\thetitle.\quad}

\usepackage{multirow,array}
\usepackage{setspace}
 \title{
 Transient impact from the Nash equilibrium of a permanent market impact game
 }

\author[1\footnote{ 
Corresponding author. Tel. +39.050.221.6302.
}]{ Francesco Cordoni}
\author[2,3]{ Fabrizio Lillo}
\affil[1]{
\begin{small}
Department of Economics, Royal Holloway University of London, \protect\\ Egham TW20 0EX, UK.
\protect\\ E-mail: francesco.cordoni@rhul.ac.uk
\end{small}}

\affil[2]{\begin{small} Dipartimento di Matematica, Universit\`a di Bologna, \protect\\ 
Piazza di Porta San Donato, 5 - 40126
Bologna (BO), Italy
\end{small}
} 
\affil[3]{
\begin{small}
Scuola Normale Superiore, Piazza dei Cavalieri, 7 - 56126 Pisa (PI), Italy \protect\\ E-mail: fabrizio.lillo@unibo.it
\end{small}}

\date{
Date Written: May 1, 2022; Last revised: \today
} 
\begin{document}

\maketitle


\vspace{-1cm}

\begin{abstract}
  A large body of empirical literature has shown that market impact of financial prices is transient.
 However, 
from a theoretical standpoint, the origin of this temporary nature is still unclear. 
We show that an implied transient impact arises from the Nash equilibrium between a directional trader and one arbitrageur in a market impact game with fixed and permanent impact.
The implied impact is the one that can be empirically inferred from the directional trader's trading profile and price reaction to order flow.
Specifically, we propose two approaches to derive the functional form of the decay kernel of the Transient Impact Model, one of the most popular empirical models for transient impact, from the behaviour of the directional trader at the Nash equilibrium. The first is based on the relationship between past order flow and future price change, while in the second we solve an inverse optimal execution problem. We show that in the first approach the implied kernel is unique, while in the second case infinite solutions exist and a linear kernel can always be inferred.

 
\end{abstract}



\noindent\textbf{Acknowledgments:} FC acknowledges the support of the projects
``How good is your model? Empirical evaluation and validation of quantitative models in economics''
funded by MIUR Progetti di Ricerca di Rilevante Interesse Nazionale (PRIN) Bando 2017 and of the Leverhulme Trust Grant Number RPG-2021-359 - ``Information Content and Dissemination in High-Frequency Trading".
The authors gratefully acknowledge Jim Gatheral and Francesco Feri for the helpful discussions and suggestions that have helped to improve the paper.

\section{Introduction}

Markets are far from being perfectly elastic and any order or trade causes prices to move, which results in a cost for traders termed market impact cost.
Therefore, in order to
minimize market impact cost, 
agents typically fragment their orders into small pieces, which are executed incrementally over the day, see, e.g., \cite{almgren2001optimal}, \cite{Moallemi}.
On the other hand, other agents might take advantage of the knowledge that a trader is purchasing a certain amount of assets progressively by buying at the beginning and selling at the end of the trader's execution using predatory trading strategies, 
\cite{Brunnermeier,Carlin}.
 All these results produce a correlation in order flows which subsequently generate and increase the market impact effect.

The relation between trades and price is known as market impact.
In the seminal work of \cite{almgren2001optimal} 
the market impact is modelled by a constant fixed over time, making market impact permanent and constant.
However, many empirical evidences \cite{Bouchaud,BFL}, \cite{Zarinelli}, \cite{Taranto}  have shown that market impact is in great part transient, i.e. the effect of order flow on prices is not constant but decays with time. One of the most popular methods to describe transient impact is the so called \emph{Transient Impact Model} (TIM) of \cite{Bouchaud,BFL}, also known as the propagator model. The TIM postulates that price is a linear combination of past order flow modulated by a decaying function of time.
Nevertheless, 
the origin of transient impact is still unclear from a theoretical point of view.

In this work, we investigate theoretically how observed transient impact can emerge as the result of the interaction between different types of agents, even when the underlying impact is fixed and permanent.
To this end we use the market impact game framework of \cite{schied2018market} which allows us to study the equilibrium 
characterizing the price dynamics in terms of the activity of two or more agents simultaneously trading, \cite{schoneborn2008trade}, \cite{schied2018market}, \cite{luo_schied} \footnote{Although market impact games suffer from Nash equilibrium oscillations, which may affect price dynamics and trigger market instability, see \cite{cordoni2020instabilities}, we set up the parameter model in such a way that these spurious oscillations are controlled and prevented.}.
In particular, we consider 
a market impact game  with {\it linear permanent} price impact between two traders who want to liquidate their position in a finite time horizon.
This market impact game model with linear permanent impact corresponds to the classical Almgren-Chriss framework generalized to two agents.
As in \cite{Moallemi}, one of the agents 
acts as an \emph{arbitrageur}, which attempts 
to make a profit by
exploiting market price 
movements caused by
the liquidation of the other agent named \emph{directional}. However, in contrast to them, we consider the symmetric information game\footnote{To be precise, in the market impact game none of the agents has private information on the fundamental value of the asset.} of \cite{schied2018market}.
The price dynamics is
obtained as the Nash equilibrium of the game.
Even if recent works have highlighted how market impact games can be studied in the more general setting with $J$ risk-averse agents trading $M$ assets, see \cite{luo_schied} and \cite{cordoni2020instabilities},
respectively, we consider the simplest setting with $J=2$ risk-neutral agents trading the same asset ($M=1$) of the original framework of \cite{schied2018market}.
We may interpret these two traders as representative agents of the market.
However, in the following we discuss also the generalization of this setting to the multi-agent case.

Once the Nash equilibrium solution of the market impact game with linear permanent impact is found, we consider it from the perspective of an external observer who looks at the price dynamics and at the execution of the directional agent and tries to estimate from this data the market impact function. More specifically, we consider two different approaches: in the first, the observer estimates the impact function from the observed relation between the trading volume of the directional agent and the price dynamics. This corresponds to the financial industry practice of estimating price impact models by regressing price realizations over past traded volumes exploiting large datasets of algorithmic executions.
In the second approach, the observer looks at the executions of the directional agent and infers which price impact he might have used for the optimal execution. This second approach has been sometimes used in the literature, see for example \cite{Zarinelli}.

In both cases, we find that the inferred impact is transient and that it is consistent with the TIM of \cite{Bouchaud,BFL}, despite the fact the equilibrium solution of the game has been obtained with a permanent impact as in \cite{almgren2001optimal}. Thus, in this setting, the transient impact is the result of how the market impact model is derived, specifically because it has been obtained by considering only part of the order flow and its relation with the price. For these reasons in the following we will term the inferred impact as {\it implied transient impact}\footnote{
To better clarify our contribution,
we remark that the purpose of this work is not 
to provide a general optimal execution model, but to exhibit 
evidence of transient impact in a suitable simple market setting, 
as described as follows.}.





{\bf Related literature.} 
Trading in financial markets can be naturally modelled as a dynamic game between agents who trade a given asset. Since trading affects price, the reward (or cost) of an agent depends on how the other trade. Thus each player is trying to anticipate and respond to the actions of the others. 
In optimal execution, if the liquidation of the large trade order is performed too quickly, other market participants may notice it and try to front-run the trade, driving up the price and potentially reducing the trader's profit. On the other hand, if the trader moves too slowly, the opportunity to execute the trade at a favourable price may vanish.
Therefore, traders must find the optimal execution strategy based on their best guess of what other market participants will do. Thus at each time the decision of a trader depends on the past actions of the other trader(s) as typically happens in a dynamic game. Moreover the process is not simply a repeated game, because the conditions change at each time and the trading decision taken at one time step affects the payoff in future time steps.\\
\indent Dynamic games are widely used in finance from investment and corporate finance problems to bankruptcy games, for an exhaustive review, see \cite{basar2018handbook}. For the above reasons, in recent years the interest toward optimal execution problems in the dynamic games literature has grown, e.g.,
\cite{MOALLEMI2012361}, \cite{Huangdgames}, \cite{dong2022stackelberg}.
Standard optimal execution algorithms \citep[e.g.,][]{almgren2001optimal} aims at minimizing market impact cost without considering the presence of other investors or, to be more precise, consider them in an aggregated and non strategic form.  \cite{Moallemi} shows that such an approach leads to optimal schedules that are quite 
unrealistic and counterproductive.  Indeed, they exhibit predictive 
behaviour which can be easily detected by arbitrageurs and therefore
they increase execution market impact cost.
Thus, in an optimal execution problem, a trader should acknowledge the simultaneous presence of other agents. 
 In \cite{Moallemi} the authors formulate the optimal execution problem as a dynamic game with asymmetric
information, with a trader and a single arbitrageur. The market impact is modelled as a linear permanent price impact
model.  The authors analyzed and computed the Bayesian equilibrium of the game numerically. \\
\indent Our setting is different, since we consider the symmetric market impact game framework of \cite{schied2018market}, where the authors show the existence and uniqueness of the related Nash equilibrium, which turns out to be deterministic with a closed-form expression. 
Moreover, this framework aligns with the optimal execution algorithm context, where a trader has to plan the liquidation schedule with an a priori strategy. 
The resulting equilibrium turns out to reduce transaction costs even when no other competitors are present during the considered optimal liquidation trading time interval in the sense of Nash equilibrium.
 In this paper, we exhibit how to relate the market impact function to the dynamic activity of agents in optimal execution problems. We show how the interaction of different traders generates an implied transient impact, even when the underlying impact is fixed and permanent.
The dynamic nature of our game is determined by the fact that 
the optimal strategy takes into account both past and future actions of the agents. To better highlight the dynamic nature of our game, in Section \ref{subsection_miope}, we also present a different version of the game where agents are assumed to be myopic, i.e. they find the Nash equilibrium sequentially at each time step, and we show that the solution is different from the one obtained from the game used in this paper.

The paper is loosely related to the ``fair pricing” theory of \cite{Waelbroeck} where an equilibrium condition is derived between liquidity providers and a broker aggregating informed orders from several funds, in which the average price paid during the execution is equal to the price at the end of the reversion phase. 
Authors connect the distribution of order size with the shape of impact trajectory and to the price reached after the end of the execution. The reversion predicted by the model is a clear sign of a partially transient nature of impact in the model. \cite{Waelbroeck} propose that the transient nature of impact is related to the fact that liquidity providers, who observe an algorithmic execution of a large trade, are uncertain whether the execution is finished or not, while in the present paper the execution horizon is fixed. Moreover \cite{Waelbroeck} do not derive the explicit form of the decay kernel.


A different modeling approach to explain the transient nature of impact is via the modeling of the Latent Limit Order Book of \cite{Donier} which assumes that each long term investor has a reservation price (to buy or to sell) that they update, due to incoming news, price changes, noise, etc. All these trading intentions constitute the latent liquidity, i.e. is not immediately posted in the public order book. When the market price hits the reservation price of a given buy (sell) investor, his order is executed. Reservation prices remain sticky during a typical memory time and impact is expected to decay as a power-law of time, reaching a small asymptotic value after times corresponding to the memory time of the market. This decay is again a sign of the transient nature of impact and, in fact, under the assumption of a small trading rate, the price dynamics in the Latent Limit Order Book coincides with the one of TIM.  Although interesting, this approach is quite different from ours, which is based on a Nash equilibrium solution between two different types of agents. 

A closer point of comparison is the recent study of \cite{vodret2020stationary}, where the authors proposed a micro-foundation for the propagator using a self-consistent equation for the propagator function derived (as a limit) by an equilibrium of an agent-based system.
However,
even if propagator like models can be seen as equilibria of suitable agent-based models, the evidence of 
\cite{vodret2020stationary} does not fully explain the typical propagator shape
of transient impact in terms of order flows derived by 
optimal schedule strategies. 

{\bf Structure of the paper.} The paper is organized as follows. In Section \ref{market_impact_game} we recall the market impact games framework and we analyze the related Nash equilibrium by showing its symmetries when the price impact is constant.
In Section \ref{sec_alternative_app} we propose the price dynamics approach to implied transient impact, whereas.
in Section 
\ref{sec_intr_tim} we show how to relate the Nash equilibrium to an equivalent optimal execution problem by presenting theoretical results which characterize the implied transient impact function.
Finally, in Section \ref{conclusion} we conclude. 
All the proofs are reported in Appendix \ref{sec_app_proof}.

\section{Market Impact Games and Transient Impact Model}\label{market_impact_game}

\subsection{The Schied and Zhang setting}
Following \cite{schied2018market}, we consider the standard framework of market impact games, i.e., 
two risk-neutral traders who want to liquidate the same asset during the same time interval $ \T=\{t_0,t_1,\ldots,t_N\},$ where
 $0=t_0<t_1<\cdots <t_N=T$.

Given a suitable probability space 
$(\Omega,(\F_t)_{t\geq0},\F,\p)$, the price dynamics is described by a right-continuous martingale, $S_t^0$, when none of the agents trade.
However, the two traders want to unwind
 a given initial position with inventory $Z\in \R$,
 where a positive (negative) inventory represents a short (long) position,
 during a given trading time grid $ \T$ and following 
 an admissible strategy, which is a sequence of random variable 
 $\bm{\zeta}=(\zeta_0,\zeta_1,\ldots,\zeta_N)$ such that
 $\zeta_k \in \F_{t_k} \text{ and bounded}\  \forall k=0,1,\ldots,N$, and
 $\zeta_0+\zeta_1+\cdots+\zeta_N=Z$, see \cite{schied2018market} for further details.
The components of $\bm{\zeta}$
represent the order flow during the trading time $t_k$ for each $k=0,1,\ldots,N$.
 We denote with $X_1$ and $X_2$ the initial 
inventories of the two considered agents, with  $\Xi=(\xi_{i,k})\in \R^{2 \times (N+1)}$ the matrix 
of the respective strategies, where 
$\bm{\xi}_{1,\cdot}=\{ \xi_{1,k} \}_{k \in \numberset{T}}$
and 
$\bm{\xi}_{2,\cdot}=\{ \xi_{2,k} \}_{k \in \numberset{T}}$
are the strategies of trader $1$ and $2$, respectively.
Thus, when the two agents place orders, the price dynamics is characterized by
the market impact. In the 
original work of \cite{schied2018market}
 it is assumed that the price impact is described by
 transient impact model of \cite{Bouchaud, BFL}, which describes the price process $S_{t}^{\Xi}$
affected by the 
strategies $\Xi$ of the two traders, i.e.,
\begin{equation}\label{eq:tim}
   S_t^{\Xi}= S_t^0 -\sum_{t_k <t} G(t-t_k)(\xi_{1,k}+\xi_{2,k}), 
   \quad \forall\ t \in \numberset{T},
\end{equation}
where $G:\R_{+}\to \R_+$ is the market impact function,  also called \emph{decay kernel}, describing the lagged price impact of a unit buy or sell order
overtime.

The objective of the agents is to minimize their expected costs given the other traders strategies, $\E[C_{\T}(\bm{\xi}_{1,\cdot}|\bm{\xi}_{2,\cdot})]$, where the 
cost function is $C_{\T}(\bm{\xi}_{1,\cdot}|\bm{\xi}_{2,\cdot})$ is described by the sum of the permanent impact and the temporary impact modeled by a quadratic term $\theta \xi_{1,k}^2$ at trading time $k$. 
More precisely,
let $(\varepsilon_i)_{i=0,1,\ldots N}$
     be an i.i.d. sequence of Bernoulli $\left(\frac{1}{2}\right)$-distributed random variables 
     that are independent of $\sigma(\bigcup_{t\geq0}\F_t)$. Then the 
     cost of $\bm{\xi}_{1,\cdot}\in \mathscr{X}(X_1,\T)$ given $\bm{\xi}_{2,\cdot}\in \mathscr{X}(X_2,\T)$
     is defined as
  \begin{equation}\label{eq:cost}
      C_{\T}(\bm{\xi}_{1,\cdot}|\bm{\xi}_{2,\cdot})=
     \sum_{k=0}^N \left(
      \frac{G(0)}{2}\xi_{1,k}^2-\Saff \xi_{1,k}+\varepsilon_k G(0) \xi_{1,k} \xi_{2,k} +
      \theta \xi_{1,k}^2
      \right)+ X_1 S_0^0,
\end{equation}
     where $\mathscr{X}(X,\T)$ is the set of admissible strategies for the initial inventory $X$ on a specified time grid $\T$. Similarly, the cost of $\bm{\xi}_{2,\cdot}$ given $\bm{\xi}_{1,\cdot} $
 is defined in analogous way.
 The Bernoulli variable $\varepsilon_k$ models the execution priority at time $t_k$, which is given to the trader who wins the independent (Bernoulli) coin toss game.
We refer to \cite{schied2018market} for a complete discussion on the definition of the cost functional. The temporary impact  
$\theta \xi_{j,k}^2$, for each trader $j$, 
models the slippage cost, even if it 
can also be interpreted 
as a quadratic transaction fee. 
In this work, we adopt the mathematical modeling of Schied and Zhang 
and we do not specify exactly what this term represents.

Since we are interested in studying the optimal strategies of the two agents, under complete and perfect information assumption, where the agents want to minimize the expected costs of their strategies, we consider the following definition of Nash equilibrium.
Given the expected costs functionals of the two agents,
the Nash Equilibrium of a market impact games
  is a pair $(\bm{\xi}_{1,\cdot}^*,\bm{\xi}_{2,\cdot}^*)$ of strategies in $\mathscr{X}(X_1,\T)\times 
  \mathscr{X}(X_2,\T)$ such that
  \[
  \begin{split}
   \E[C_{\T}(\bm{\xi}_{1,\cdot}^*|\bm{\xi}_{2,\cdot}^*)]&=\min_{\bm{\xi}_{1,\cdot}\in \mathscr{X}(X_1,\T)}
   \E[C_{\T}(\bm{\xi}_{1,\cdot}|\bm{\xi}_{2,\cdot}^*)] \text{ and }\\
   \E[C_{\T}(\bm{\xi}_{2,\cdot}^*|\bm{\xi}_{1,\cdot}^*)]&=\min_{\bm{\xi}_{2,\cdot}\in \mathscr{X}(X_2,\T)}
   \E[C_{\T}(\bm{\xi}_{2,\cdot}|\bm{\xi}_{1,\cdot}^*)].
   \end{split}
  \]
Then, \cite{schied2018market} showed that for any strictly positive definite (in the sense of Bochner) decay kernel $G$, time grid $\T$, transaction cost parameter $\theta \geq0$, and initial inventories $X_1$, $X_2$
there exists a unique Nash equilibrium 
$(\bm{\xi}_{1,\cdot}^*,\bm{\xi}_{2,\cdot}^*)$
and it is deterministic. Moreover, it is provided by
\begin{align}\label{eq_1_S&Z}
\bm{\xi}_{1,\cdot}^*&=\frac{1}{2} (X_1 +X_2)\bm{v}
+\frac{1}{2} (X_1 -X_2)\bm{w}
\\
\label{eq_2_S&Z}
\bm{\xi}_{2,\cdot}^*&=\frac{1}{2} (X_1 +X_2)\bm{v}
-\frac{1}{2} (X_1-X_2)\bm{w},
\end{align}
where the fundamental solutions $\bm{v}$ and $\bm{w}$ are 
defined as 
$
   \bm{v}=\frac{1}{\bm{e}^T (\Gamma_{\theta}+\widetilde{\Gamma})^{-1}\bm{e}}(\Gamma_{\theta}+\widetilde{\Gamma})^{-1}\bm{e}$ and
     $\bm{w}=\frac{1}{\bm{e}^T (\Gamma_{\theta}-\widetilde{\Gamma})^{-1}\bm{e}}(\Gamma_{\theta}-\widetilde{\Gamma})^{-1}\bm{e}
$
and  $\bm{e}=(1,\ldots,1)^T \in \R^{N+1}$. 
The solutions are called fundamentals, since they 
represent the Nash equilibrium when $X_1=X_2=1$ and 
when $X_1=-X_2=1$, respectively.
The kernel matrix $\Gamma \in \R^{(N+1)\times (N+1)}$ is given by
  $
   \Gamma_{ij}=G(|t_{i-1}-t_{j-1}|), \ i,j=1,2,\ldots,N+1,
 $
 for $\theta\geq0$, $
   \Gamma_{\theta}:=\Gamma+2\theta I$,   and the matrix $\widetilde{\Gamma}$ is given by
  \[
\widetilde{\Gamma}_{ij}=
  \begin{cases}
   \Gamma_{ij} & \text{ if }i>j\\
      \frac{1}{2}G(0) & \text{ if }i=j,\\
      0 & \text{ otherwise.}
  \end{cases}
  \]
For the sake of simplicity, we refer to the previous framework as ``Schied and Zhang market-impact game".

\subsection{Market impact games with constant impact }\label{sec_ne_primary_model}

The above setting cannot be used when $G(t)\equiv G_1\in\R_+$ $\forall t$, i.e. in a
\cite{almgren2001optimal} setting. In fact, the assumption of strictly 
positive definite (in the sense of Bochner) of $t\mapsto G(|t|)$
no longer holds and the kernel matrix $\Gamma=G \bm{e}\bm{e}^T \in 
\R^{(N+1)\times (N+1)}$ is singular.
  However, the existence and uniqueness of Nash Equilibrium 
  associated with the two agents ($\bm{\xi}^*$, $\bm{\eta}^*$)
    is guaranteed when 
   $\Gamma_{\theta}=\Gamma+2\theta I$ is definite positive,   and the matrices $\Gamma_{\theta}\pm\widetilde{\Gamma}$ are invertible, 
     see proof of Theorem 1 and Lemma 3 and 4 of 
  \cite{schied2018market}.
   Thus, by the matrix determinant lemma\footnote{If $A\in\R^{N\times N}$ is a non singular square matrix and $\bm{u},\bm{v}\in \R^N$, then 
$\det(A+\bm{uv}^T)=(1+\bm{v}^T A^{-1} \bm{u})\det(A)$. 
}, see also Theorem 1 of \cite{ding2007eigenvalues}, the eigenvalues of $\Gamma_{\theta}$
are $
\lambda_1=G \cdot (N+1)+2\theta\text{ and }\lambda_{2:(N+1)}=2\theta,$
so it is definite positive as long as $\theta>0$.
$\Gamma_{\theta}-\widetilde{\Gamma}=\widetilde{\Gamma}^T+2\theta I$ is 
an upper triangular matrix where its diagonal elements are different from zero, i.e., it is 
non singular and finally $\Gamma_{\theta}+\widetilde{\Gamma}  $
is non singular if $\theta>0$, see Appendix 
\ref{app_1}.
Therefore, provided that $\theta>0$, we proved that also when $G(t)=G>0$ is constant the Nash equilibrium exists and it is 
given by the previous equations \eqref{eq_1_S&Z} and 
\eqref{eq_2_S&Z}. 

In the previous framework, we have described the market impact game model, i.e., a Schied and Zhang market-impact game with a constant impact function.
Since the market impact function $G_1\in\R_+$ is constant, without loss of generality, we may fix $G_1= 1$ and
in order to prevent market instability\footnote{The instability appears as a result of oscillating Nash equilibria, which in turn affect price dynamics, see \cite{cordoni2020instabilities} for further details. 
However, if $\theta\geq G(0)/4$ these spurious oscillations disappear, see \cite{schied2018market}.} we 
set $\theta\geq G_1/4$, e.g., $\theta=1$.

We consider the interaction between a directional seller and
an arbitrageur, where without loss of generality, we may assume that their inventory is given by
$X_{direc}=1$ and $X_{arb}=0$, respectively. 
Then, the Nash equilibrium for the directional is given by 
the average of the fundamental vectors $\bm{v}$ and $\bm{w}$.
Interestingly both the optimal strategies are symmetric in time.
All the proofs are given in Appendix \ref{sec_app_proof}.

    \begin{pr}\label{direc_solut_symm}
    Let $\theta>0$, then the fundamental solutions, $\bm{v}$ and $\bm{w}$, of a Schied and Zhang market-impact game where the market impact function $G_1$ is constant, are equal up to a time-symmetry, i.e., 
    \begin{equation}
\bm{v}_{k}=\bm{w}_{N+2-k}, \ k=1,2,\cdots,N+1.
\end{equation}
Furthermore, if we denote $(\Gamma_{\theta}+\widetilde{\Gamma})=A$,
  \begin{equation}\label{explicit_fund_solut}
  \bm{v}=
  \frac{
 A^{-T}\bm{e}}
{\bm{e}^TA^{-T}\bm{e}}, \quad
     \bm{w}=
     \frac{
  A^{-1}\bm{e}}
{\bm{e}^TA^{-1}\bm{e}},
\end{equation}
and $v_1=\frac{1}{\lambda \cdot (1-a^{N+1})}$ and 
$v_n=\frac{a^{n-1}}{\lambda \cdot (1-a^{N+1})}$ for $n=2,\ldots,N+1$,
where $\lambda=2\theta/G+\frac{1}{2}$ and $a=1-1/\lambda.$

    \end{pr}

In the standard single-agent \cite{almgren2001optimal} framework, 
the optimal 
schedule for the directional is to trade with a constant rate over the trading periods, making the optimal execution independent from the permanent impact.
However, 
when in the market impact game
we specify a constant market impact function, the optimal solution for the directional provided by the Nash equilibrium has a U-shape, although it is assumed a permanent impact model as in
\cite{almgren2001optimal}. On the other hand, for the arbitrageur, the optimal schedule 
has a round-trip shape, see Figure \ref{fig_ne_shape}.

\begin{te}\label{NE_symm}
In a Schied and Zhang market-impact game where the two agents are a directional and an arbitrageur,
then the Nash Equilibrium 
is given by the following strategies:
\begin{align}\label{eq_1_S&Z_te}
\bm{\xi}_{direc,\cdot}^*&=\frac{1}{2} X_{direc}( \bm{v}
+\bm{w}),
\\
\label{eq_2_S&Z_te}
\bm{\xi}_{arbi,\cdot}^*&=\frac{1}{2} X_{direc}( \bm{v}
-\bm{w}),
\end{align}
where $\bm{v}$ and $\bm{w}$ are the fundamental solutions and $X_{direc}$ is the inventory of the directional agent.
Moreover, if 
the market impact function is constant and $\theta>0$, the Nash equilibrium is time-symmetric, i.e., 
\begin{align}\label{eq_symm_NE_onedim}
{\xi}_{direc,k}^*&={\xi}_{direc,N+2-k}^*, \ k=1,2,\ldots,N+1, \\
{\xi}_{arbi,k}^*&=-{\xi}_{arbi,N+2-k}^*, \ k=1,2,\ldots,N+1. 
\label{eq_symm_NE_2_onedim}
\end{align}
\end{te}

Specifically, 
even if the Nash equilibrium of the directional agent is time-symmetric when $\theta>0$ in the market-impact games with constant market impact function, when
$\theta>\theta^*=G_1/4$ it is also positive, strictly decreasing in the first $\lfloor{(N+1)/2}\rfloor$ components and convex, i.e., it has a U-shape.
 
\begin{co}\label{cor_Ushape}
In a Schied and Zhang market-impact game where the two agents are a directional, with inventory $X_{direc}>0$, and an arbitrageur, the market impact $G_1$ is constant
and $\theta>\theta^*=G_1/4$, the Nash equilibrium of the directional, $\bm{\xi}_{direc,\cdot}^*$, has a U-shape, i.e., it is time-symmetric, positive, strictly decreasing in the first $\lfloor{(N+1)/2}\rfloor$ components and convex, where $N+1$ is the number of the trading time step.
\end{co}

      \begin{figure}[t]
\centering
{\includegraphics[width=1\textwidth]{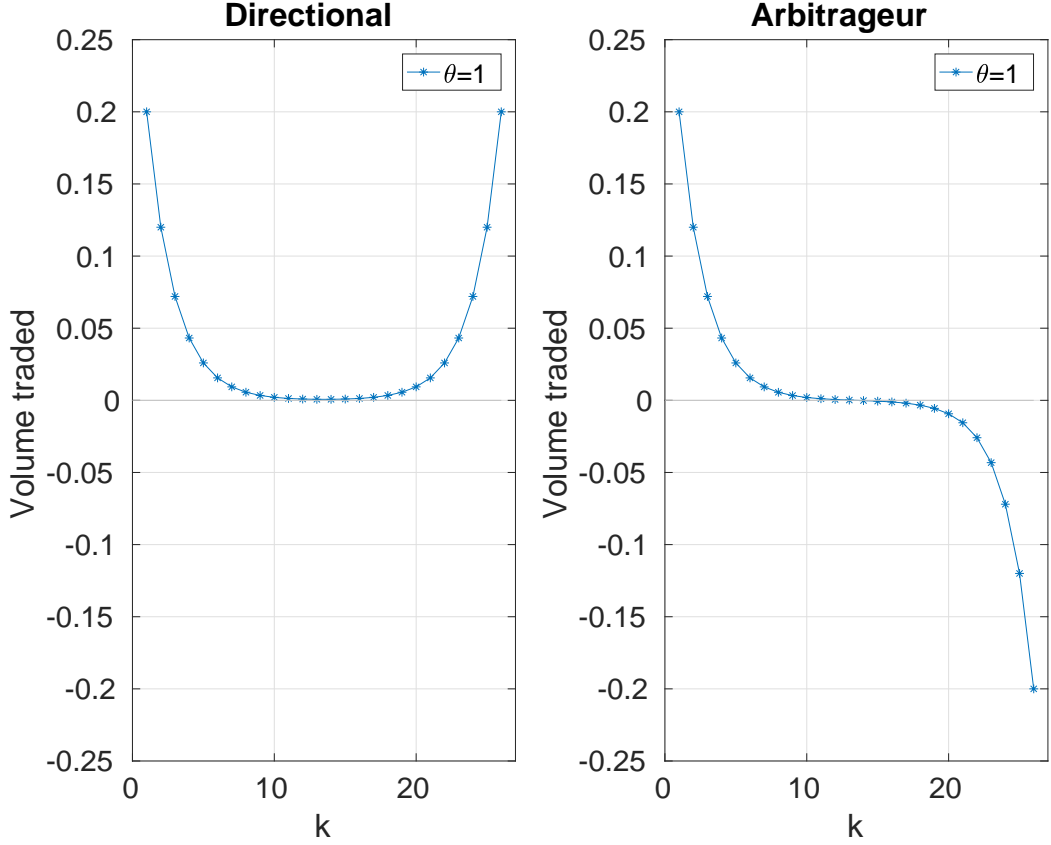}}
\caption{ Nash equilibrium $\bm{\xi}_{direc}^*$ of the directional (left)
and $\bm{\xi}_{arbi}^*$ of the arbitrageur (right) trading only one asset. The trading time grid is equidistant with 26 points and 
$\theta=1$. The market impact is constant $G_1=1.$} 
\label{fig_ne_shape}
\end{figure}

\subsubsection{Myopic market-impact game}
\label{subsection_miope}

To highlight the dynamic nature of the our game, we present here a different model where agents optimize their trading by finding the Nash equilibrium at each time interval, i.e. without considering the effect of their action on {\it future} prices (and rewards). For this reason we term the agents as myopic. 
This new game might be thought of as a repeated game problem, even if the game will still remain dynamic, due to the fact that, as we will see, the optimal actions depend on past price. We show that the solution of the myopic market-impact game is different from the fully dynamical Schied-Zhang game.

For the sake of simplicity, we focus on the case where both agents are fundamentalist and identical.
Then, if we denote by $S_{t_k}^{\bm{\xi},\bm{\eta}}$ the price at the beginning of the interval and assume a constant\footnote{The computation is straightforwardly extended to the case of a variable $G(t)$.} $G$, the cost function for agent $1$ is (see Eq. \ref{eq:cost})
\[
C(\xi_{1,k}|\xi_{2,k}) = 
      \frac{G}{2}\xi_{1,k}^2-S_{t_k}^{\bm{\xi},\bm{\eta}} \xi_{1,k}+\varepsilon_k G \xi_{1,k} \xi_{2,k} +
      \theta \xi_{1,k}^2.
\]
whose expectation is 
\[
\E[C(\xi_{1,k}|\xi_{2,k})] = \frac{G}{2}\xi_{1,k}^2-S_{t_k}^{\bm{\xi},\bm{\eta}} \xi_{1,k}+\frac{G}{2} \xi_{1,k} \xi_{2,k} +
      \theta \xi_{1,k}^2.
\]
Therefore, the best response function of agents $1$ and $2$ is:
\[
\begin{split}
\xi^{brf}_{1,k} &= \arg \min \E[C(\xi_{1,k}|\xi_{2,k})] = 
\frac{S_{t_k}^{\bm{\xi},\bm{\eta}}-\frac{G}{2} \xi_{2,k}}{G+2\theta}, \\    
\xi^{brf}_{2,k} &= \arg \min \E[C(\xi_{2,k}|\xi_{1,k})] = 
\frac{S_{t_k}^{\bm{\xi},\bm{\eta}}-\frac{G}{2} \xi_{1,k}}{G+2\theta}.
\end{split}
\]
Thus, considering 
\[
\xi^{brf}_{1,k} = \arg \min \E[C(\xi_{1,k}|\xi_{2,k})] = 
\frac{S_{t_k}^{\bm{\xi},\bm{\eta}}-\frac{G}{2} \xi^{brf}_{2,k}}{G+2\theta}
\]
we may recover the related Nash equilibrium at time $k$ for both agents
    \[
    \xi^{**}_{1,k} =\xi^{**}_{2,k}= \frac{2S_{t_k}^{\bm{\xi},\bm{\eta}}}{3G+4\theta} \equiv  \alpha S_{t_k}^{\bm{\xi},\bm{\eta}}.
    \]
    Therefore, the Nash equilibrium depends only on the price at the beginning of each interval. Notice that we have not set
any type of constraint on the inventory that each agent wants to liquidate.

The price increment in interval $k$ is
\[
S_{t_{k+1}}^{\bm{\xi},\bm{\eta}} - S_{t_{k}}^{\bm{\xi},\bm{\eta}} = - G(\xi_{1,k}^{**} +\xi_{2,k}^{**}) = -2\alpha G S_{t_{k}}^{\bm{\xi},\bm{\eta}}.
\]
thus $S_{t_k}^{\bm{\xi},\bm{\eta}} = (1-2\alpha G)^k S_0$.  If $0 < 1-2\alpha G < 1$ the price decays exponentially fast until the full inventory has been liquidated. Clearly $1-2\alpha G < 1$ because  $\alpha$ and $G$ are positive. Interestingly, $ 1-2\alpha G>0 $ avoid price oscillations, and this condition is satisfied if $\theta>G/4$, i.e., the stability condition of Schied and Zhang market impact games, see \cite{schied2018market} and \cite{cordoni2020instabilities}.

Finally, the number of time intervals is set by the condition, $X_1 = \sum_{k=1}^{N+1} \xi_{1,k}^{**}$
which gives $N = \frac{\log(1- 2 X_1 G/ S_0)}{\log(1-2\alpha G)}$, so that the number of trading rounds is fixed at the beginning, as in the Schied and Zhang game.

The solution of the myopic game is an admissible strategy for the Schied and Zhang market impact game. 
However, using Proposition \ref{direc_solut_symm}, a direct inspection shows that the Nash equilibrium of the latter is different from the solution of the former. Specifically, at the first interval it is $ \xi_{1,1}^{**}> \xi_{1,1}^{*}$, which means that the myopic traders prefer to liquidate more at the beginning. Finally, by computing the average expected cost of the whole execution, the Nash Equilibrium of the myopic version is in general suboptimal, i.e. providing a larger cost (for both traders) than the fully dynamic Schied and Zhang impact game.
\subsection{Optimal execution in the Transient Impact Model}\label{sec:timopt}

We now recall how to derive the optimal execution schedule in the standard TIM. In this case, the equation of price, similarly to Eq. (\ref{eq:tim}), is 
$$
   S_t^{\bm{\eta}}= S_t^0 -\sum_{t_k <t} G_2(t-t_k)~\eta_{k}, 
   \quad \forall\ t \in \numberset{T},
$$
where, to avoid confusion with the market impact game, we denote with $\bm{\eta}$ the trading strategy. Only one agent trade and thus there is no explicit interactions with other agents.

It is possible to show \citep{BFL,schied2018market} that the expected cost of the directional agent is 
$\E[C_T (\bm{\eta})]=\frac{1}{2} \bm{\eta}^T \Gamma_{\theta,2} \bm{\eta}$,
where $\Gamma_{\theta,2}$ is the decay matrix defined above and corresponding to the kernel $G_{2}(t)$.
 Minimizing the expected cost, the optimal solution for a directional trader with inventory $X_{direc}$ is obtained by 
\[
\bm{\eta}_{direc}^{*}=\frac{X_{direc}}{\bm{e}^T\Gamma_{\theta,2}^{-1} \bm{e} }
\Gamma_{\theta,2}^{-1} \bm{e}.
\]
Since $\Gamma_{\theta,2}$ is symmetric then it is trivial that 
$\bm{\eta}_{direc}^{*}$ is also time symmetric as $\bm{\xi}^*_{direc}$.
Moreover, it is straightforward that if we assume 
a constant market impact function $G_2$ the solution of $\bm{\eta}^*_{direc}$ is constant in time as for the classical \cite{almgren2001optimal} solution. 
In the general case, where $G_2(t)$ is a strictly positive decay kernel, 
the optimal execution is characterized by a U-shape, e.g., see Figure \ref{fig_ne_shape_TIM}.

   \begin{figure}[t]
\centering
{\includegraphics[width=0.75\textwidth]{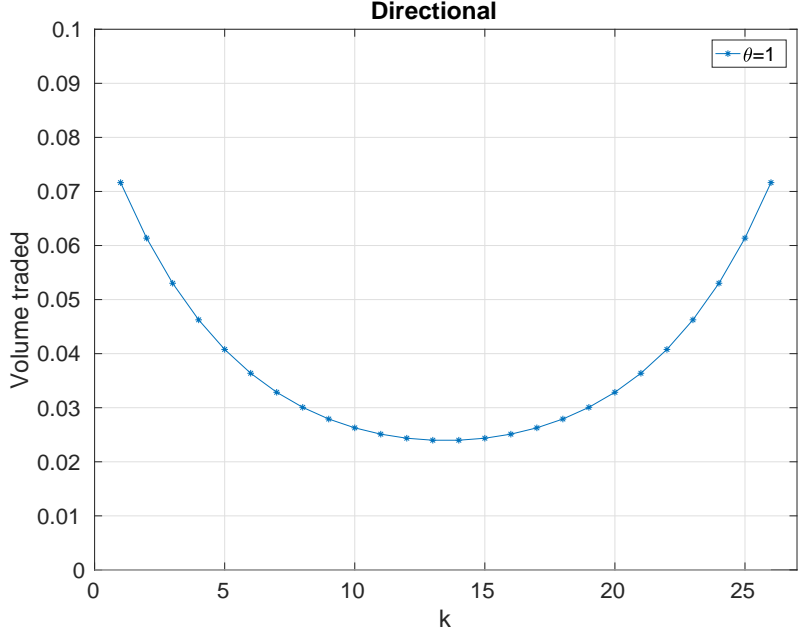}}
\caption{Optimal execution schedule $\bm{\eta}_{direc}^*$ of a directional agent in a TIM model. 
The trading time grid is equidistant with 26 points and 
$\theta=1$. The market impact function is an exponential decay kernel $G_2(t)=\exp(-t).$} 
\label{fig_ne_shape_TIM}
\end{figure}





\section{Implied Transient Impact via 
Price Dynamics}
\label{sec_alternative_app}

We first present a direct approach to 
obtain evidence of transient market impact, 
by looking at the (average) price dynamics obtained in  the market impact game
described in Section \ref{sec_ne_primary_model}.
Remind that the price dynamics on a given trading time grid 
$\T=\{t_0,t_1,\ldots,t_N\}$ is described by
\[
S_t^{\Xi}=S_t^0-\sum_{t_k <t} G(t-t_k)(\xi_{1,k}+\xi_{2,k}),
\]
where $S_t^0$ is a right-continuous martingale defined on a given probability space which acts as volatility term.
Therefore, if we discard the noise due to volatility, the above relation holds
between the expected value of price increments and order flow.
Thus, if $S_0=S_0^0$,
\[
\overline{S}_t^{\Xi}-S_0=-\sum_{t_k <t} G(t-t_k)(\xi_{1,k}+\xi_{2,k}),
\]
where $\overline{X}$ denotes the expectation of $X$, which can be recasted in a matrix form,
\begin{equation}\label{eq_alternative_system}
    \bm{S}=-C \bm{\Xi},
\end{equation}
where $\bm{S}=(\overline{S}_{t_1}^{\Xi}-S_0,\ldots, \overline{S}_{t_{N+1}}^{\Xi}-S_0)$
 is the aggregate drift, $\bm{\Xi}=\bm{\xi}_1+\bm{\xi}_2=
(\Xi_1,\ldots,\Xi_{N+1})^T$ is the aggregate order flow and 
\[
C=
\begin{bmatrix}
G(t_1-t_0) & 0 & 0 & \cdots  &\cdots   &0\\
G(t_2-t_0) & G(t_2-t_1) & 0 & 0&  \cdots  &0\\
G(t_3-t_0) & G(t_3-t_1) & G(t_3-t_2) & 0& \ddots & 0\\
\vdots & \ddots &\ddots  & \ddots & \ddots &\vdots\\
G(t_{N}-t_0) & G(t_{N}-t_1) &\cdots  &\cdots   & G(t_{N}-t_{N-1}) & 0  \\
G(t_{N+1}-t_0) & G(t_{N+1}-t_1) & G(t_{N+1}-t_2)&   \cdots &  \cdots& G(t_{N+1}-t_{N})  \\
\end{bmatrix}.
\]
Notice that
we may rewrite the above system as $\bm{S}=-\mathcal{M} \bm{g}$, where, in the case that $\T$ is an equidistant time grid,
$\bm{g}=(G(t_0),G(t_1),\ldots,G(t_{N}))^T$
and 
\[
\mathcal{M}=
\begin{bmatrix}
\Xi_1 & 0 & 0 & \cdots  &\cdots   &0\\
\Xi_2 & \Xi_1 & 0 & 0&  \cdots  &0\\
\Xi_3 & \Xi_2 & \Xi_1 & 0& \ddots & 0\\
\vdots & \ddots &\ddots  & \ddots & \ddots &\vdots\\
\Xi_N & \Xi_{N-1} &\cdots  &\cdots   & \Xi_1 & 0  \\
\Xi_{N+1} & \Xi_N & \Xi_{N-1}&   \cdots &  \cdots& \Xi_1  \\
\end{bmatrix}.
\]
If the aggregate net order flow is different from zero at $t_0$, i.e., $\Xi_1\neq0$, the matrix $\mathcal{M}$ is non singular
and therefore there is always an unique solution
\begin{equation}\label{eq_defi_alter_intrinsic_decay_kernel}
    \bm{g}=-\mathcal{M}^{-1}\bm{S}
 \end{equation}
from which we can recover the kernel $G(t)$.


Thus, one could use $\bm{\xi}_1$ and $\bm{\xi}_2$ , the Nash equilibrium solution of Schied and Zhang market impact game and obtain $\bm{S}$, from Equation \eqref{eq_alternative_system}. Then, assuming only
one agent (the directional, say agent 1) one solves Equation \eqref{eq_defi_alter_intrinsic_decay_kernel} using the order flow of the directional, i.e., $\Xi= \bm{\xi}_1$,
so that we can obtain the implied decay kernel associated with the single agent TIM
characterised by the price dynamics of the market impact games model. This procedure to infer the (transient) impact model corresponds to the usual practice, in the academia and in the financial industry, to estimate market impact from large sets of algorithmic executions by regressing price changes over past traded volumes. Equation \eqref{eq_defi_alter_intrinsic_decay_kernel} leads
to a first definition of \emph{implied transient impact function},
which we emphasise as:


\begin{de}\label{intrinsic_def_2}
(Implied  transient  impact  function -  Price Approach). The  transient  impact  function $G^{(P)}_{impl}(t)$  which satisfies 
$\bm{S}=-\mathcal{M} \bm{g}$, where $\mathcal{M}$ depends on $\bm{\xi}_{direc}$ and $\bm{S}$ is recovered by the aggregate (drift) order flows, $\Xi$, of the market impact game  is  called \emph{implied transient impact function}.
\end{de}

The implied decay is uniquely identified. 
We emphasise these results in the following theorem.

\begin{te}\label{te_alter_defin}
If the aggregate net order flow is different from zero at $t_0$, i.e., $\Xi_1\neq0$, the linear system $\bm{S}=-\mathcal{M} \bm{g}$ has unique solution for $\bm{g}$.
\end{te}

As a specific first example we consider the market impact  model where
the directional and the arbitrageur trade on an 
equidistant time grid $\T_N=\{\frac{kT}{N}|k=0,1,\ldots N\}$ where $T=1$,
$N=25$, $\theta=1$, with inventory equal to $1$ and $0$, respectively.
The decay kernel is set to $G_1\equiv 1$.
Then, we compute the cumulative drift $\bm{S}$ generated by the interaction of the two agents in the market impact games.
We solve the system \eqref{eq_defi_alter_intrinsic_decay_kernel} where $\mathcal{M}$ is computed by considering only the drift generated by
the directional  $\bm{\xi}_{direc}$ and we report in the left panel of Figure \ref{fig:alternativeSZ} the implied transient impact function.

Before commenting on this figure, we note that this approach can be easily extended to the case when more arbitrageurs are present. Specifically, we compute $G^{(P)}_{impl}(t)$ in a market impact game with one directional trader and two arbitrageurs, following the \cite{luo_schied} model, where the decay kernel is again set to $G_1\equiv1$. 
We recall that the Nash Equilibrium, in this case, is not anymore time-symmetric, see Figure \ref{fig_3_agents_NE} below. The right panel of Figure \ref{fig:alternativeSZ} shows the implied transient impact function in this three player game.

Figure \ref{fig:alternativeSZ} indicates that the implied market impact $G^{(P)}_{impl}(t)$ is transient and nonlinear in all settings.
Qualitatively, we observe that the implied transient impact kernels differ in terms of ``size impact", i.e., the absolute value of $G^{(P)}_{impl}(0)$ is greater for the five/three agent game than those obtained starting from the two agent game, 
but they exhibit the same shape and they are both decreasing functions. 
Therefore, to compare them fairly, we have divided each $G^{(P)}_{impl}(t)$ by the related value at zero, $G^{(P)}_{impl}(0)$.
The initial values for the (non-scaled) implied market impact functions were $2$, $3$ and $5$, for the two, three and five players game, respectively. We observe that the implied market impact function has a more sharp decline when the number of arbitrageurs increases.
We remark that in order to account the way by which impact depends on the number of agents, a proper scaling factor should be applied on both implied kernels, as the one discussed in \begin{NoHyper}\cite{cordoni2020instabilities}\end{NoHyper}. However, for our purpose is sufficient to equally compare both implied market functions as done in Figure \ref{fig:alternativeSZ}.


\begin{figure}
    \centering
    \subfloat[][]
    {\includegraphics[width=0.8\textwidth]{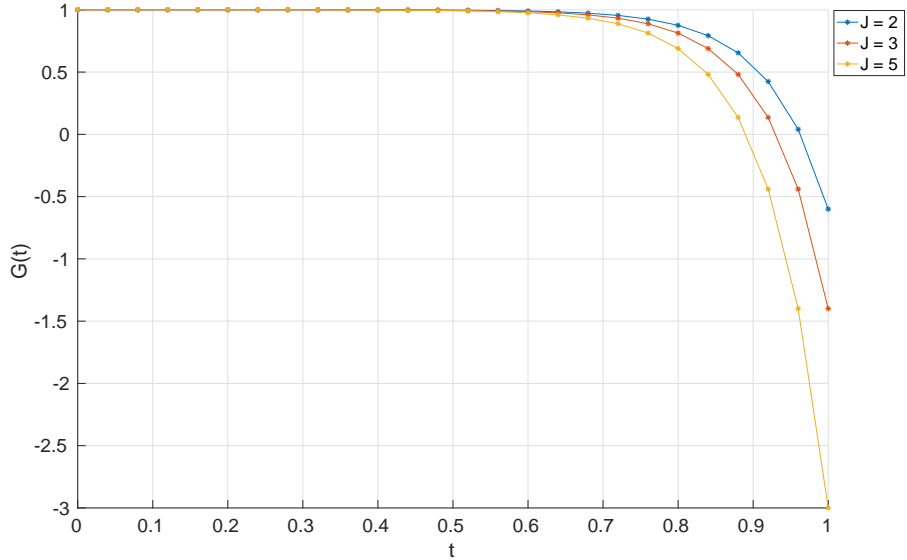}}
    \caption{Scaled implied transient impact function, computed with the price approach, solving the system \eqref{eq_defi_alter_intrinsic_decay_kernel}. The blue, red and orange lines correspond to the scaled $G^{(P)}_{impl}(t)$ obtained from a game with one directional and one arbitrageur,
    from a game with one directional and two arbitrageurs, and from a game with one directional and four arbitrageurs, respectively. Each function was scaled in such a way its initial value is equal to one. }
    \label{fig:alternativeSZ}
\end{figure}

In conclusion, following this approach, we have found evidence of the transient nature of market impact. However, the implied impact results to be a concave function, in contrast to what many empirical studies have found, e.g., \cite{Bouchaud}.
In Appendix \ref{app_altern_approach2} we investigate how we can fix this undesired feature of the implied transient impact function, by solving Eq. \eqref{eq_defi_alter_intrinsic_decay_kernel} employing the scheduling of Almgren and Chriss optimal execution model.
However, the motivation of this alternative solution is not straightforward.

In the next section we propose a different approach to recover the implied market impact function by solving an inverse optimal execution problem.



\section{Implied Transient Impact via 
Optimal Execution}
\label{sec_intr_tim}

In the second approach we propose to derive the implied transient impact from the optimal execution schedule. As in the previous section, we assume that the actual decay kernel is a constant $G_1\in \R_+$. As shown in Section \ref{sec_ne_primary_model}, the optimal solution
of the directional trader at the Nash equilibrium $\bm{\xi}^*_{direc}$
has a U-shape.
As discussed in section \ref{sec:timopt}, a U-shape is also exhibited 
 in the optimal execution problem using the transient impact model.  
Therefore, 
when a suitable transient impact function is selected,
the U-shape of $\bm{\xi}^*_{direc}$ is equivalent to optimal schedules obtained by single-agent transient impact models. 
In other words, given the optimal solution $\bm{\xi}_{direc}^*$
it is possible to select an appropriate propagator function 
so that $\bm{\xi}^*_{direc}=\bm{\eta}^*_{direc}$.
We denote this implied transient impact function as $G_{impl}(t)$,
More precisely we define:
\begin{de}[Implied transient impact function - Optimal execution approach]\label{def_intrinsic_decay_kernel}
The transient impact function $G^{(OE)}_{impl}(t)$ such that 
the optimal schedule obtained by a TIM with a single agent, $\bm{\eta}^*_{direc}$, is equal to the Nash equilibrium of the directional trader in the market impact game, $\bm{\xi}^*_{direc}$, is called \emph{implied transient impact function}.
\end{de}

Thus the question we plan to answer is: given the 
solution obtained by a market impact game, how is it possible 
to derive the corresponding implied decay kernel associated with the single-agent TIM?
Moreover, is the implied transient impact function unique?


Without loss of generality we assume $X_{direc}=1$.
Given $\bm{\xi}_{direc}^{*}\in \R^{N+1}$ we ask whether the equation 
\begin{equation}\label{eq_probl_solver}
  \bm{\xi}_{direc}^{*}=\frac{\Gamma_{\theta}^{-1} \bm{e}}{\bm{e}^T\Gamma_{\theta}^{-1} \bm{e} }
\end{equation}
has solution and 
if it is unique. 
Since $\theta>0$ is given, the only unknown is the symmetric Toeplitz matrix
$\Gamma$ such that $\Gamma_{\theta}=\Gamma+2\theta I$, which depends on $N+1$ parameters,
thus in principle there are $N+1$ equations in $N+1$ unknowns. However, it is clear that if $\Gamma$ is a solution,
then any $\Gamma+K\bm{e} \bm{e}^T$, where $K\in \R-\{\frac{-1}{\bm{e}^T\Gamma_{\theta}^{-1}\bm{e}}\}$ is a solution, as observed in the following remark.

\begin{os}\label{up_constant}
Using Sherman-Morrison formula\footnote{
If $\Gamma$ is a solution of \eqref{eq_probl_solver}, $\Gamma_{\theta}^{-1}$ is positive i.e., $\bm{x}^T\Gamma_{\theta}^{-1}\bm{x}>0\ \forall \bm{x}\in \R^{N+1}$, see e.g. Lemma 2 of \cite{schied2018market}. Therefore $1+K \bm{e}^T\Gamma_{\theta}^{-1}\bm{e}\neq0$ if and only if $K\neq \frac{-1}{\bm{e}^T\Gamma_{\theta}^{-1}\bm{e}} $.}, it is straightforward that
$$ (\Gamma_{\theta}+K \bm{e} \bm{e}^T)^{-1}\bm{e}=
\Gamma_{\theta}^{-1} \bm{e}-
\frac{K\Gamma_{\theta}^{-1} \bm{e} \bm{e}^T\Gamma_{\theta}^{-1}\bm{e}}{1+K\bm{e}^T\Gamma_{\theta}^{-1}\bm{e}}
=
\left(1-
\frac{K  \bm{e}^T\Gamma_{\theta}^{-1}\bm{e}}{1+K\bm{e}^T\Gamma_{\theta}^{-1}\bm{e}}\right)
\Gamma_{\theta}^{-1} \bm{e}
=\frac{\Gamma_{\theta}^{-1}\bm{e} }{1+K\bm{e}^T\Gamma_{\theta}^{-1}\bm{e}},$$
for each $K \in \R$. Then,
$$
\bm{e}^T(\Gamma_{\theta}+K \bm{e} \bm{e}^T)^{-1}\bm{e}=
\frac{\bm{e}^T\Gamma_{\theta}^{-1}\bm{e} }{1+K\bm{e}^T\Gamma_{\theta}^{-1}\bm{e}},
$$ thus,
$$
\frac{(\Gamma_{\theta}+K \bm{e} \bm{e}^T)^{-1}\bm{e}}{\bm{e}^T(\Gamma_{\theta}+K\bm{e} \bm{e}^T)^{-1}\bm{e}}=
\frac{\Gamma_{\theta}^{-1}\bm{e} }{ \bm{e}^T\Gamma_{\theta}^{-1}\bm{e}}.
$$
So $G(t)$ and $G(t)+K$, where $K\in \R-\{\frac{-1}{\bm{e}^T\Gamma_{\theta}^{-1}\bm{e}}\}$,
generate the same optimal execution of a TIM, then the decay kernel is identifiable up to a constant.
\end{os}

Moreover,
$\Gamma_{\theta}$ is identified up to a multiplicative constant, i.e., if $\Gamma_{\theta}$ satisfies \eqref{eq_probl_solver}
then any $\alpha \Gamma_{\theta}$ where $\alpha\neq0$ is a solution.
Therefore, we may set $G(0)$ such that the elements of the main diagonal of $\Gamma_{\theta}$, which are $G(0)+2\theta$, are equal to $1$.
 However,
 even if we select a particular class of decay kernel
 we show that in general the identification of $G^{(OE)}_{impl}(t)$
is not related only to a constant and multiplicative scaling.

To solve problem \eqref{eq_probl_solver}, let us set
$$
\Pi^{-1}=\frac{\Gamma_{\theta}^{-1}}{\bm{e}^T\Gamma_{\theta}^{-1} \bm{e} }
$$
hence $$
\Pi=(\bm{e}^T\Gamma_{\theta}^{-1} \bm{e} )\Gamma_{\theta}$$
and we may rewrite Equation \eqref{eq_probl_solver} as
\begin{equation}\label{eq_probl_solver2}
\Pi \bm{\xi}_{direc}^{*} =\bm{e}  
\end{equation}
where the unknowns are the $N+1$ different entries of 
$\Pi=\mbox{Toep}(g_0,g_1,\ldots,g_{N})=\mbox{Toep}(\bm{g})$.
Moreover $\Pi$ is symmetric and satisfies $\bm{e}^T \Pi^{-1} \bm{e}=1$.
 Thus, we recast Equation \eqref{eq_probl_solver2}
 into a linear system in the unknowns $\bm{g}$. 
 Let us consider as an example the case of $N+1=4$.
 The original system is formulated as 
 $$
 \begin{bmatrix}
 g_0 & g_1 & g_2 & g_3\\
  g_1 & g_0 & g_1 & g_2\\
   g_2 & g_1 & g_0 & g_1\\
    g_3 & g_2 & g_1 & g_0\\
 \end{bmatrix}
 \begin{bmatrix}
 \xi_1 \\
 \xi_2 \\
   \xi_3 \\
    \xi_4 \\
 \end{bmatrix}=\begin{bmatrix}
 1 \\ 1\\ 1\\1
 \end{bmatrix},
 $$
 where we omit the upper $*$ for the sake of simplicity,
 which can be rewritten as
  $$
 \begin{bmatrix}
 \xi_1 & \xi_2 & \xi_3 & \xi_4\\
  \xi_2 & \xi_1+\xi_3 & \xi_4 & 0\\
  \xi_3 & \xi_2+\xi_4 & \xi_1 & 0\\
 \xi_4 & \xi_3 & \xi_2 & \xi_1\\
 \end{bmatrix}
 \begin{bmatrix}
g_0 \\
 g_1 \\
   g_2 \\
    g_3 \\
 \end{bmatrix}=\begin{bmatrix}
 1 \\ 1\\ 1\\1
 \end{bmatrix},
 $$
 i.e., as 
 \begin{equation}\label{eq_with_H}
     H\bm{g}=\bm{e}
 \end{equation}
 from which we may recover $\bm{g}$. 
 However, since $\bm{\xi}^*_{direc}$ has a U-shape, in particular it is time-symmetric, Equation \eqref{eq_with_H} has in general 
 infinite solutions. This means that the implied transient 
 impact function can not be identified without imposing some restrictions.
 Before
 showing the main results we provide two examples.
 \begin{os}
 If for some reasons $\bm{\xi}^*_{direc}$ is not time-symmetric (for example it is the result of the Nash equilibrium of a
directional trader against $M > 1$ arbitrageurs, see Section \ref{sec_multi_agent}) and there are no other symmetries, the matrix $H$ is
full rank and invertible. The unique solution is
 the constant vector $\bm{g}=\frac{\bm{e}}{\bm{e}^T \bm{\xi}^*_{direc}}$, which provides however a singular matrix $\Pi$.
 Therefore, as we expect, when $\bm{\xi}^*_{direc}$ is not time-symmetric 
 there is no $\Gamma$ which satisfies Eq. \eqref{eq_probl_solver}.
 \end{os}
 
 \begin{es}
 Let us consider the case when $N+1=4$ and let us 
 suppose $\bm{\xi}^*_{direc}$ be a U-shaped.
 As a consequence the matrix $H$ is not full rank.
We recover $\bm{\xi}^*_{direc}$ by the optimal execution of a TIM where
 $\Gamma_{\theta}= \mbox{Toep}(1, 0.6, 0.5, 0.2)$, 
 thus $\bm{\xi}^*_{direc} = 12^{-1}(5, 1, 1, 5)$. The rank of the matrix $H$ is 2, so the space of the solutions is infinite and it has dimension $2$ and can be 
 parametrized as 
 $\bm{g}=\left(\frac{60}{29}-\frac{\alpha}{29}-\frac{30}{29}\beta,
 \frac{48}{29}-\frac{24}{29}\alpha+\frac{5}{29}\beta,\alpha,\beta\right)$.
 Choosing $\alpha=\beta=0$ we obtain $\Pi=\mbox{Toep}(60/29,48/29,0,0)$, which is clearly not proportional to the ``original" $\Gamma_{\theta}$, which can be recovered
 with 
 $\alpha=0.8450704$, $\beta = 0.3380282$, getting $\Pi=1.6901408\Gamma_{\theta}$.
 \end{es}

 \begin{te}[Identification Problem of the Implied Transient Impact Function]\label{id_intrinsic_decay}
 Let us suppose that $\bm{\xi}^*_{direc}\in \R^{N+1}$ has a U-shape, 
 then the system 
 \begin{equation}\label{eq_probl_solver2_te}
\Pi \bm{\xi}_{direc}^{*} =\bm{e} \Longleftrightarrow
H\bm{g}=\bm{e}
\end{equation}
where $\Pi=\mbox{Toep}(\bm{g})$ is symmetric
and $\bm{g}\in \R^{N+1}$ has strictly decreasing components, 
 has infinite solutions, where the rank of $H$ is equal to $(N+1)/2$ if $N+1$ is even and $N/2+1$ if $N+1$ is odd, respectively.
 \end{te}
 In other words, in general, there are $(N+1)/2$ and $N/2$ solutions
 as much as implied transient impact functions, when $N+1$ is even and odd, respectively.

\subsection{The Implied Linear Transient Impact Function}
\label{sec_linear_TIM}
We now investigate the previous problem when the implied 
decay kernel is restricted to be linear.
We assume that the time grid is equidistant 
$\T_N=\{\frac{kT}{N}|k=0,1,\ldots,N\}$, where withouth loss of generality $T=1$ and $N\in \N$.
If $G^{(OE)}_{impl}(t)=\alpha+\beta t$, where $\beta<0$ and, since the
kernel is identified up to a constant term, we impose 
without loss of generality that $\alpha$
 is such that $G^{(OE)}_{impl}(T)=0$. The previous matrix $\Pi$
 is proportional up to a constant to the decay kernel matrix $\Gamma_{\theta}$, so we search for a 
 $\Pi=\mbox{Toep}(g_0,g_1,\ldots,g_N)$, where 
 $g_k=2\theta\delta_{k,0}+\alpha+\beta\frac{k}{N}$, for $k=0,1,\ldots,N$ where $\delta_{k,0}$ is $1$ for $k=0$ and 0 otherwise.
 Then, if $X_{direc}$ is the inventory of the directional agent, since $\Pi \bm{\xi}_{direc}^*=1$ we may recover\footnote{For the sake of simplicity we remove the upper $*$ and $direc$ from the component of $\xi^*_{direc}$.} the first $\lfloor{N/2+1}\rfloor$
 conditions:
 \begin{small}
 \begin{enumerate}[align=left]
     \item[(Eq. 1)] $\alpha X_{direc}+2\theta\xi_1+\frac{\beta}{N} \sum_{i=2}^{N+1}(i-1) \xi_i=1 $;
     \item[(Eq. 2)]  $\alpha X_{direc}+2\theta\xi_2+\frac{\beta}{N}  \left(\xi_1+\sum_{i=3}^{N+1}(i-2) \xi_i\right)=1 $; 
     \item[(Eq. $k$)]  $\alpha X_{direc}+2\theta\xi_3+\frac{\beta}{N} \left((k-1)\xi_1+(k-2)\xi_2+\cdots+2\xi_{k-2}+\xi_{k-1}+
     \sum_{i=k+1}^{N+1}(i-k) \xi_i\right)=1 $, 
 \end{enumerate}
 \end{small}
 where $k\leq\lfloor{N/2+1}\rfloor$. 
  From Corollary \ref{cor_Ushape} if $\theta>\theta^*=G_1/4$, the components of $\bm{\xi}$ 
 are all positive.
 Thus, subtracting each equation from the previous one we obtain that
 \begin{enumerate}[align=left]
     \item[(Eq. 1) $-$ (Eq. 2)]: $2\theta (\xi_1-\xi_2)
     +\frac{\beta}{N} \left( 
     \sum_{i=2}^{N+1} \xi_i-\xi_1
     \right)=0
     $;
     \item[(Eq. $k-1$) $-$ (Eq. $k$)]:   $2\theta (\xi_{k-1}-\xi_k)
     +\frac{\beta}{N} \left( 
     \sum_{i=k}^{N+1} \xi_i-     \sum_{i=1}^{k-1} \xi_i
     \right)=0$, where $k\leq\lfloor{N/2+1}\rfloor$.  
 \end{enumerate}
 Therefore,
 since $X_{direc}=\sum_{i=1}^{N+1}\xi_i$,
 we may compute $\beta$ using the previous equation (Eq. 1) $-$ (Eq. 2),
 \[
 \beta=\frac{-2\theta N \cdot (\xi_1-\xi_2)}{X_{direc}-2\xi_1},
 \]
 but from (Eq. $k-1$) $-$ (Eq. $k$)
 \[
 \beta=\frac{-2\theta N \cdot (\xi_{k-1}-\xi_k)}{X_{direc}-2\sum_{i=1}^{k-1}\xi_i},
 \quad k\leq\lfloor{N/2+1}\rfloor.
 \]
 Therefore, it must hold for  $k\leq\lfloor{N/2+1}\rfloor$
 \begin{equation}\label{cond_linear_decay_first}
     \frac{-2\theta N \cdot (\xi_1-\xi_2)}{X_{direc}-2\xi_1}=
     \frac{-2\theta N \cdot (\xi_2-\xi_3)}{X_{direc}-2\xi_1-2\xi_2}=
     \cdots=
     \frac{-2\theta N \cdot (\xi_{k-1}-\xi_k)}{X_{direc}-2\sum_{i=1}^{k-1}\xi_i},
 \end{equation}
 which, since $\theta \neq 0$,  is equivalent to
  \begin{equation}\label{cond_linear_decay}
     \frac{ (\xi_1-\xi_2)}{X_{direc}-2\xi_1}=
     \frac{(\xi_2-\xi_3)}{X_{direc}-2\xi_1-2\xi_2}=
     \cdots=
     \frac{ (\xi_{k-1}-\xi_k)}{X_{direc}-2\sum_{i=1}^{k-1}\xi_i},\
     k\leq\lfloor{N/2+1}\rfloor
     .
 \end{equation}

 Then, we have proven
 that if the components of $\bm{\xi}^*_{direc}$ satisfy the relations  \eqref{cond_linear_decay}
 then $\Pi=\mbox{Toep}(\bm{g})$, where 
 $g_k=2\theta\delta_{k,0}+\alpha+\beta\frac{k}{N}$, $k=0,1,\ldots,N$ is a solution for \eqref{eq_probl_solver2_te}.
 Actually, the above conditions on $\bm{\xi}^*_{direc}$ are a necessary and sufficient conditions for the existence of linear implied transient impact function in the above model setting.
 We now may prove the following result.

 
 
\begin{te}[Linear Implied Transient Impact Function]\label{te_linear_beta_results}
 Let  $\T_N=\{\frac{kT}{N}|N\in \N, \ k=0,1,\ldots,N\}$, 
 be an equidistant time grid and $\theta>\theta^*=G_1/4$. There always exists a linear implied transient impact described by $\Pi=\mbox{Toep}(\bm{g})$, where 
 $g_k=2\theta\delta_{k,0}+\alpha+\beta\frac{k}{N}$, $k=0,1,\ldots,N$ and $\beta= \frac{-2\theta N \cdot (\xi_{1}-\xi_2)}{X_{direc}-2\xi_1}$.
 \end{te}
 
 The theorem tells that linear solutions of $\eqref{eq_probl_solver}$ are contained in a one dimensional affine space, so 
 if we impose that $G^{(OE)}_{impl}(T)=0$, then we may identify
 the unique implied linear transient impact function.
 The slope is provided by 
 $-\theta N \frac{(1-a)^2}{a}=
 -\frac{4\theta N G_1^2}{16\theta^2-G_1^2}
 $, since $a=1-1/\lambda$,
 $\lambda=2\theta/G_1+\frac{1}{2}$. 
 Since $\theta > G_1/4$ then the slope is always negative, i.e., the linear implied market impact 
 function is effectively a decay kernel.
 We also observe that this slope is in absolute value an increasing function of $G_1$ and it decreases with $\theta$. We further investigate this latter relation 
 in Section \ref{transaction_cost_effect}

\begin{os}\label{os_two_theta}
 We observe that if $\theta=0$, then the relations \eqref{cond_linear_decay_first}
 are satisfied. 
 However, the Nash Equilibrium of the market impact game model exists and it has no oscillations when $\theta\geq\theta^*>0$.
 Therefore, we may set two different $\theta$ one for the market impact game, $\theta_1$, such that the Nash equilibrium is well defined and one for the optimal execution model $\theta_2$ where we set $\theta_2=0$. We discuss this particular case in Section \ref{sec_transaction_cost_two_theta}.
 \end{os}

\subsection{Role of Transaction Costs }\label{transaction_cost_effect}
 
In general, the implied transient impact depends on the level of transaction costs. This is due to the fact that in the market impact game, the Nash equilibrium depends on the parameter $\theta$.
To give a concrete example, let us consider the same setting as in 
Section \ref{market_impact_game} and let us focus on the implied impact via optimal execution, restricting our attention to linear functions.
In Section \ref{sec_linear_TIM} we proved that the slope of the implied impact is  $ -\frac{4\theta N G_1^2}{16\theta^2-G_1^2}$. So, the absolute value of the slope is a decreasing function of $\theta$ and it goes to zero for large values of $\theta$.
This phenomenon can also be explained by looking at the interaction between the two agents in the market impact game.
Indeed, when transactions costs increase, the interaction between agents in market impact game disappears, since the arbitrageur
gradually reduces the traded volume, so that the optimal
schedule of the directional becomes the same as in the standard Almgren-Chriss framework, see Figure \ref{fig_sol_AC_theta}.
Therefore, when $\theta$ increases, there is less and less interaction between the agents and the implied transient impact function 
vanishes.

     \begin{figure}[!t]
\centering
{\includegraphics[width=1.05\textwidth]{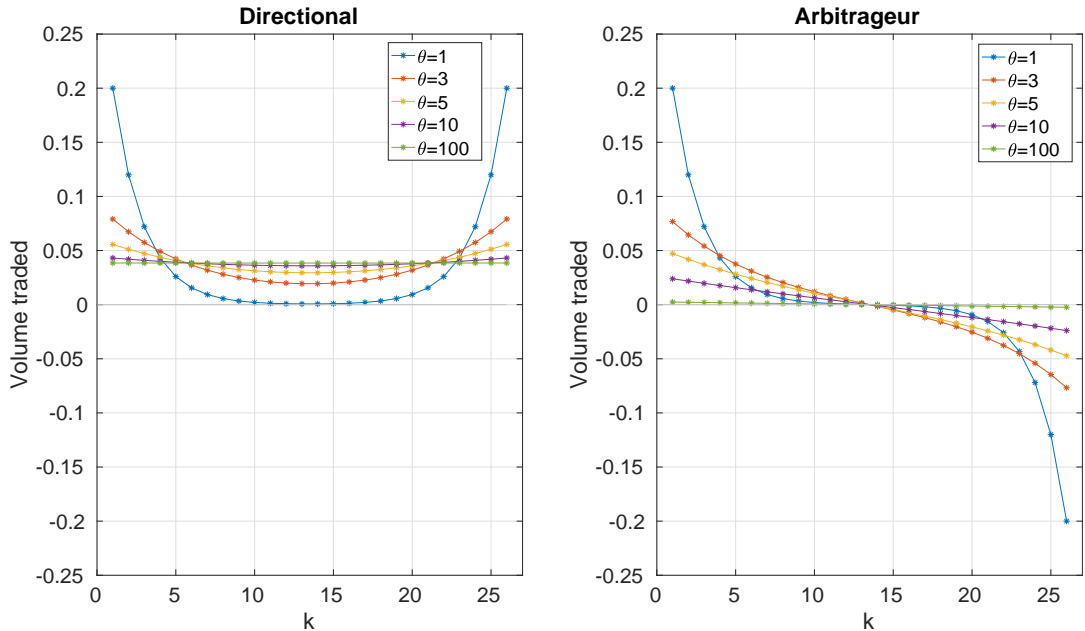}}
\caption{Nash equilibria of market impact games with fixed $G_1=1$, and equidistant time grid where $T=1, N=25$, when the transaction costs level $\theta$ increases.} 
\label{fig_sol_AC_theta}
\end{figure}

\subsection{The Multi-Agent Case}\label{sec_multi_agent}

Even if, in the previous setting, we may interpret  the two traders as representative agents, we now analyze the general setting with more than two agents. In a multi-asset market impact game with $J>2$ agents, we use the results of \cite{luo_schied} and \cite{cordoni2020instabilities} to derive the Nash equilibria of agents. However, the results of Section \ref{sec_ne_primary_model} are no more valid 
when we consider $J>2$ traders,
 since the fundamental solutions depend on the number of agents, which implies that the decay kernel matrix is no more Toeplitz in general, see e.g., \cite{luo_schied}. 
As an example, we consider $2$ arbitrageurs and a directional trader\footnote{All the agents are assumed to be risk-neutral.}, which trade the same asset.
Figure \ref{fig_3_agents_NE}
exhibits the related Nash equilibria for the agents when 
we set $T=1$, $N=25$, $G_1=1$, $\theta=1$. 
We observe that the solutions of the two arbitrageurs are identical.
Even if for the arbitrageur the optimal solution is 
always a round-trip strategy facing the same direction of the directional trader at the beginning of the session,
it is quite evident that the optimal solution for the directional is a U-shape which is no more time-symmetric. In particular, it is optimal for him/her 
to trade more at the end of session, exploiting the arbitrageur's impact.  

\begin{figure}
    \centering
    \includegraphics[scale=0.4]{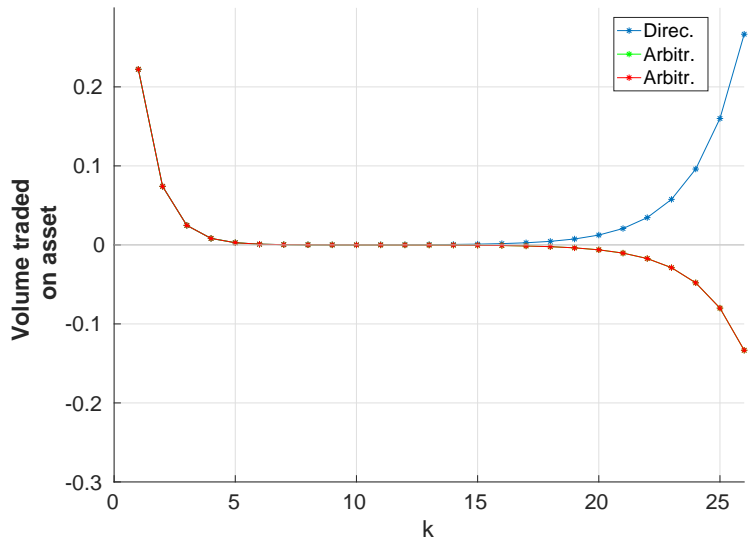}
    \caption{Nash equilibria for the market impact game model with a directional trader and two arbitrageur, where 
    $T=1$, $N=25$, $G_1=1$, $\theta=1$. The solution of the two arbitrageurs are identical.}
    \label{fig_3_agents_NE}
\end{figure}

\subsection{Is it possible to incorporate transaction costs in the decay kernel?}
\label{sec_transaction_cost_two_theta}

In this section we examine the relations between the transaction cost level $\theta$s and the implied transient impact function.
%
 Specifically, we generalize the above setting by considering two different $\theta$s for the two models and
 by asking whether it is possible to select a kernel function for the TIM optimal execution problem in such a way the kernel incorporates the transaction cost level without estimating it.
We focus our attention on the case when 
$\theta_1\geq G_1/4$, while $\theta_2=0$ for the optimal execution model. This last condition implies that
$\bm{\eta}_{direc}^*$ is characterized by the vector $\Gamma^{-1} \bm{e}$ 
and the implied transient impact function $G^{(OE)}_{impl}(t)$ is such that it incorporates the interaction between the two agents together with the 
 transaction costs of the market impact games, without relying on the transaction cost level.

In Section \ref{sec_linear_TIM} we have shown that there exists a linear implied transient impact function, which is unique up to constant, where the slope is uniquely identified. 
However, when $\theta_1 = G_1/4$ but $\theta_2=0$ any linear decay kernel (regardless of the slope coefficient) can be selected as an implied market impact function.
When $\theta_1\neq G_1/4$ and $\theta_2=0$ there are no linear implied transient impact solution.
To show that, we consider a different 
perspective when solving the optimal execution inverse problem.
We start by choosing a specific kernel $G_2(t)$ in the optimal execution model and we search for a suitable parameter setting of the market impact game model so that the selected decay kernel are the corresponding implied impact function, i.e., $G_2(t)\equiv G^{(OE)}_{impl}(t)$.
 In other words, given $\bm{\eta}^*_{direc}$ we search a suitable parameter setting for the market impact game such that $\bm{\xi}^*_{direc}=\bm{\eta}^*_{direc}$.
$G_2(t)$ is specified to be linear.


If $G_{2}(t)= \alpha+ \beta t$, where $\beta\neq0$, then for an equidistant time grid, $t_i=i/N, \ 
i=0,1,\ldots,N$,
$\Gamma_{i,j}=\frac{\beta|i-j|}{N}+\alpha$. The inverse of $\Gamma$
is given by, see
\cite{dow2002explicit},
\[
\Gamma^{-1}=
\frac{1}{2\beta}
\begin{bmatrix}
-\eta_{N}/\eta_{N+1} & 1 & 0 &\cdots & 0 & \beta^2 /\eta_{N+1} \\
1& -2 & 1 &0 &\cdots &0 \\
0& 1 & -2 &1 & 0 & 0 \\
& & \ddots & \ddots & \ddots & \\
0&\cdots & 0 &1 &-2 &1\\
\beta^2 /\eta_{N+1} & 0 &\cdots &0 &1 &-\eta_{N}/\eta_{N+1}
\end{bmatrix}
\]
where $\eta_N=2\alpha\beta+\beta^2(N-1)$. Then, 
$\bm{\eta}_{direc}^*=\frac{X_{direc}}{\bm{e}^T\Gamma^{-1} \bm{e} }
\Gamma^{-1} \bm{e}=[X_{direc}/2,0,\cdots,0 , X_{direc}/2]^T$ 
regardless of any choice of $\beta$ and $\alpha$.

 Therefore, when $\theta_2=0$ in the optimal execution model, the optimal schedules are  given by a vector that concentrates the orders at the 
 two extremes of the trading session and it has a (zero) constant trading rate for all intermediate trading times\footnote{The same trading profile schedule is also exhibited in the exponential decay kernel case, where the trading velocity at intermediary trading times is constant but different 
 from zero.
 }. So, the question is whether the solution of 
 a market impact game has a shape that looks like the one described above for the single-agent TIM.   


%


\begin{es}
  \label{es_theta0}
 We observe that when\footnote{We remark that when  $\theta_1=\theta_1^*=G(0)/4$, and other generic assumptions, 
 the continuous time
Nash equilibrium of a general market impact game exists and it coincides with the high-frequency limits ($N\to \infty$) of the discrete-time
equilibrium,  as showed in \cite{schied2017high}.}
  $\theta_1=G_1/4$,  
 $\Gamma_{\theta_1}=2\theta_1 I+ G_1 \bm{e}\bm{e}^T=
 \frac{G_1}{2} I+G_1\bm{e}\bm{e}^T$ and so 

 \[ 
 (\Gamma_{\theta_1}+\widetilde{\Gamma})=
 G_1\begin{bmatrix}
 2 & 1 & 1 & \cdots &1\\
 2 & 2 & 1 &\cdots &1\\
 \vdots & & \ddots & &\vdots\\
  2 & & \cdots & 2 &1  \\
 2 & &\cdots &\cdots & 2  \\
 \end{bmatrix}, 
  (\Gamma_{\theta_1}+\widetilde{\Gamma})^{-1}=
G_1^{-1}\begin{bmatrix}
 1 &   &  & & -1/2 \\
 -1 & 1 &  \\
   & & \ddots & \ddots& \\
   & &  &-1 & 1  \\
 \end{bmatrix}
 \]
 
  \[ 
 (\Gamma_{\theta_1}-\widetilde{\Gamma})=
 G_1\begin{bmatrix}
 1 & 1 &  & \cdots &1\\
   & 1 &  &\cdots &1\\
  & &  &\ddots &\vdots\\
   & & & & 1  \\
 \end{bmatrix}, 
  (\Gamma_{\theta_1}-\widetilde{\Gamma})^{-1}=
G_1^{-1}\begin{bmatrix}
 1 &  -1 &   \\
  &  \ddots  &\ddots \\
   & &  1 &-1 \\
   & &   & 1  \\
 \end{bmatrix}.
 \]
 Therefore, $\bm{v}=[1, 0, \cdots, 0]^T$ and $\bm{w}=[0, \cdots ,0 ,1]^T$
 and the solution for the directional $\bm{\xi}_{direc}^*$ in the market impact game model
 is exactly equal to the one obtained by the optimal execution model with no transaction cost and linear decay kernel, i.e., 
 $$\bm{\xi}_{direc}^*=[X_{direc}/2, 0,\cdots,0 , X_{direc}/2]^T.$$
 \end{es}
 
Thus, when $\theta_1=\theta_1^*$, any linear decay kernel $G_2 (t)$ can be selected as 
the implied transient impact function $G_{impl}^{(OE)}(t)$ in the market impact game model.
Furthermore, it holds the following results.

\begin{pr}\label{prop_xi_spikes}
In the market impact game described in Section \ref{sec_ne_primary_model},
$$\bm{\xi}_{direc}^*=[X_{direc}/2,0,\cdots,0 , X_{direc}/2]^T,$$ if and only if $\theta_1=G_1/4$.
\end{pr}

Therefore, when $\theta_2=0$ the implied transient impact function is linear if and only if $\theta_1=G_1/4$, and all linear decay kernel functions are  
solution of the implied decay kernel problem, i.e., $\bm{\xi}_{direc}^*=\bm{\eta}_{direc}^*$ for all $G_{impl}^{(OE)}(t)$ linear.
We remark that,
contrary to Section \ref{sec_linear_TIM}, in this setting we may select any slope coefficient for $G_{impl}^{(OE)}(t)$.





\section{Conclusion}
\label{conclusion}

Understanding the transient nature of market impact is essential as it describes how prices react to trades and it is related to the information content of a trade. Rather than postulating its existence, this paper contributes to the recent literature by providing an explanation for its origin. 
We showed that transient impact naturally emerges from the Nash equilibrium of a market impact game with permanent and fixed impact. Using the setting of market impact games, our paper indicates that, in general, the impact function describing in the model the effect of trade volume on price is different from the impact function that can be inferred by an external observer who measures it from the trading activity of a specific agent (the directional in our setting). We term this inferred impact as ``implied", and we show that the implied impact of a permanent market impact game is transient.     


More specifically, we propose two approaches to derive implied impact. The first considers the lagged correlation between the directional trading volume and price changes, while the second one considers the execution of the directional trader as optimal with respect to a transient impact model and derives the possible kernel (or propagator) function. Although the implied impact is transient in both cases, there are substantial differences. In the first case, the solution is unique, while in the second one an infinite number of possible solutions. In particular, under mild assumptions on the parameters of the market impact game, a linear solution can be derived and characterized
    in terms of the Nash equilibrium. 
We also analyze the sensitivity of the implied transient impact function to transaction costs level.
    Since the implied decay kernel results from the interaction between a directional and an arbitrageur trader, when we increase the transaction costs parameter, the volume of the arbitrageur reduces until the implied transient impact no longer exists.
    Finally, we show that when we extend our framework allowing to incorporate the transaction costs in the implied price impact function, any linear family can be selected as implied market impact.
    
    In conclusion, our paper shows a possible origin of the transient nature of market impact and highlights an essential difference between the real impact ruling the game and the one that can be measured with statistical methods from trade data.

    As a possible extension, one may ask whether the evidence of transient impact under the notion of implied market impact function might be extended in a continuous-time setting. However, the continuous-time extension of the Schied and Zhang market impact game model is well defined only for $\theta = G(0)/4$, where  \cite{schied2017high} have shown that the continuous Nash equilibrium exists and it coincides with the high-frequency limit of the discrete model.
    Therefore, only in this special case the same evidence in the continuous-time model might be found. 
    The constraining assumption on $\theta$ makes the continuous-time model of marginal interest from an economic perspective. 
    On the other hand, 
    the discrete-time market impact game turns out to be more flexible and relevant. Moreover, as remarked in  
    \cite{strehle2017single}, ``continuous trading is an idealization" and every continuous-time trading strategy has to be discretized via block trades in order to be executed.
Another way to employ continuous-time modelling, which is not affected by constraining conditions on $\theta$ of continuous-time market impact game, is to
    follow \cite{strehle2017optimal}, where a different approach to modelling transaction costs is presented, so that we may avoid the singularity presented by the parameter $\theta$.
    One could solve the optimal execution problem with Fredholm integral equation, but only in the exponential case a closed-form solution may be derived, and in the general case, the equation must to solved numerically. 
 An interesting question is to analyze whether the discretization of this different model could be led back to one of the approaches we have analyzed, where we might expect to solve a discretization of a Fredholm integral equation.
 However, this further aspect is beyond the scope of this work and we leave it for further research and study.
   

\section*{Declarations}

\noindent\textbf{Ethical Approval}:
Not applicable.

\noindent\textbf{Competing interests}:
Not applicable.
 
\noindent\textbf{Authors' contributions}:
Not applicable.

\noindent\textbf{Funding}:
Not applicable.

\noindent\textbf{Availability of data and materials}:
Not applicable.

    \bibliography{bib2}
\begin{appendices}
\titlelabel{\appendixname\ \thetitle.\quad}
\appendixtitleon
\section{When  $\Gamma_{\theta}+\widetilde{\Gamma}$
is singular?}
\label{app_1}
Let $G(t)\equiv G>0$ and $\theta>0$. Then,
$\Gamma_{\theta}+\widetilde{\Gamma}=2\theta I+
\widetilde{\Gamma}+\Gamma=
A+G \bm{e} \bm{e}^T$
where 
\[
A=G
\begin{bmatrix}
\lambda \\
1 & \lambda \\
\vdots & \ddots & \ddots \\
1 & \cdots & 1 & \lambda
\end{bmatrix}, 
\quad \lambda=2\frac{\theta}{G} +\frac{1}{2},
\]
is non singular since $\lambda>0$. Moreover, it is a straightforward computation to verify that
\[
A^{-1}=\frac{1}{G}
\begin{bmatrix}
\frac{1}{\lambda} \\
-\frac{1}{\lambda^2} & \frac{1}{\lambda} \\
-\frac{(\lambda-1)}{\lambda^3} & -\frac{1}{\lambda^2} & \frac{1}{\lambda} \\
\vdots & \vdots & \ddots &\ddots \\
-\frac{(\lambda-1)^{N-2}}{\lambda^{N }} &
 -\frac{(\lambda-1)^{N-3}}{\lambda^{N-1}}
&\cdots &-\frac{1}{\lambda^2}  & \frac{1}{\lambda}
\\
-\frac{(\lambda-1)^{N-1}}{\lambda^{N+1}} &
 -\frac{(\lambda-1)^{N-2}}{\lambda^N}
&\cdots &\cdots &-\frac{1}{\lambda^2}  & \frac{1}{\lambda}
\end{bmatrix}. 
\]
Then, since $A$ is non singular, for the matrix determinant lemma 
$A+G \bm{e} \bm{e}^T$ is non singular  if and only if 
$1+G\cdot \bm{e}^TA^{-1} \bm{e}\neq0$. Let
$\bm{x}=GA^{-1} \bm{e}$, where 
$\bm{x}_1=1/\lambda$ and $\bm{x}_n=
\frac{1}{\lambda}-\frac{1}{\lambda^2}\sum_{k=0}^{n-2}
\left(\frac{\lambda-1}{\lambda}\right)^{k}
$, $n=2,\ldots,N+1$. We first observe that when $\lambda=1$,
$A+G\bm{e}\bm{e}^T$ is non singular, since
$1+G\bm{e}^T A^{-1} \bm{e}=1+G\cdot \bm{e}^TA^{-1} \bm{e}=2$. 

Thus we assume that $\lambda\neq1$.
Then if $1+G\cdot \bm{e}^TA^{-1} \bm{e}>\frac{1}{\lambda}>0$ the matrix 
$A+G\bm{e}\bm{e}^T$ is non singular.
In particular if $\theta>0$ then $1+G\cdot \bm{e}^TA^{-1} \bm{e}>\frac{1}{\lambda}$. Indeed, $1+G\cdot \bm{e}^TA^{-1} \bm{e}>\frac{1}{\lambda}$
if and only if
\[
1+\frac{1}{\lambda}+
\sum_{n=2}^{N+1} \left( 
\frac{1}{\lambda}-\frac{1}{\lambda^2}
\sum_{k=0}^{n-2} \left(
\frac{\lambda-1}{\lambda}
\right)^k
\right)>\frac{1}{\lambda}\]
\[\iff
\frac{N+1}{\lambda}-\frac{1}{\lambda}-\frac{1}{\lambda^2}
\sum_{n=2}^{N+1}\left(
\frac{1-a^{n-1}}{1-a}\right)>-1, 
\quad a=\frac{\lambda-1}{\lambda}\neq0, \mbox{ and } a \neq1
\]
\[
\iff
\frac{N+1}{\lambda}-\frac{1}{\lambda^2}
\sum_{n=1}^{N+1}\left(
\frac{1-a^{n-1}}{1-a}\right)>-1,\
\mbox{ since }
\frac{1}{\lambda}=\frac{1}{\lambda^2}\cdot \frac{1}{1-a}
\]
\[
\iff\frac{N+1}{\lambda}+1>
\frac{1}{\lambda}\sum_{n=1}^{N+1}
(1-a^{n-1})\iff
1 >-\frac{1}{\lambda}
\left(
\frac{1-a^{N+1}}{1-a}
\right)
\]
\[
\iff 1>
( 
a^{N+1}-1
) \iff 
a^{N+1}<2\iff 
a<2^{\frac{1}{N+1}}
\]

\[\iff
1-\frac{1}{\lambda}< 
 2^{\frac{1}{N+1}}
  \iff
  1-2^{\frac{1}{N+1}}<
\frac{1}{\lambda}
\]
since $\lambda=2\frac{\theta}{G}+\frac{1}{2}>0$ and $1-2^{\frac{1}{N+1}}<0$
\[
\iff\lambda>
 \left( 1-2^{\frac{1}{N+1}}\right)^{-1}, 
\]
\[
\iff
\theta>
\left(
\frac{1}{1-2^{\frac{1}{N+1}}}-\frac{1}{2}
\right)\frac{G}{2}.
\]
However, $\left(
\frac{1}{1-2^{\frac{1}{N+1}}}-\frac{1}{2}
\right)\frac{G}{2}<0$, therefore if $\theta>0$, $\theta>\left(
\frac{1}{1-2^{\frac{1}{N+1}}}-\frac{1}{2}
\right)\frac{G}{2}$ and so 
$1+G\cdot \bm{e}^TA^{-1} \bm{e}\neq0$.
Thus, if  $\theta>0$ and $G>0$ then $\Gamma_{\theta}+\widetilde{\Gamma}$ is non singular.

\section{Proofs of the results.}
\label{sec_app_proof}

\begin{lem}\label{lem_inv_gamma}
Let $\theta>0$, then the inverse of the following matrices
$$(\Gamma_{\theta}-\widetilde{\Gamma}) =
G_1\begin{bmatrix}
\frac{2\theta}{G_1}+\frac{1}{2} & 1 & 1
&\cdots  &1 &1 \\
0&\frac{2\theta}{G_1}+\frac{1}{2} & 1   &\cdots &1 &1\\
0 &\ddots&\ddots &\ddots &\ddots  &\vdots  \\
\vdots &\ddots&\ddots &\ddots &\ddots  &\vdots  \\
\vdots &\ddots&\ddots &\ddots &\frac{2\theta}{G_1}+\frac{1}{2}  & 1  \\
0&\cdots &\cdots&\cdots &0 &\frac{2\theta}{G_1}+\frac{1}{2}  
\\
\end{bmatrix}
$$
$$(\Gamma_{\theta}+\widetilde{\Gamma}) =
G_1\begin{bmatrix}
\frac{2\theta}{G_1}+\frac{3}{2} & 1 & 1
&\cdots  &1 &1 \\
2&\frac{2\theta}{G_1}+\frac{3}{2} & 1   &\cdots &1 &1\\
2 &\ddots&\ddots &\ddots &\ddots  &\vdots  \\
\vdots &\ddots&\ddots &\ddots &\ddots  &\vdots  \\
\vdots &\ddots&\ddots &\ddots &\frac{2\theta}{G_1}+\frac{3}{2}  & 1  \\
2&\cdots &\cdots&\cdots &2 &\frac{2\theta}{G_1}+\frac{3}{2}  
\\
\end{bmatrix}
$$
are given by the matrices,
$$
(\Gamma_{\theta}-\widetilde{\Gamma})^{-1}=\frac{1}{G_1} 
\begin{bmatrix}
\frac{1}{\lambda}&
-\frac{1}{\lambda^2}&
-\frac{\lambda-1}{\lambda^3}&\cdots&
 -\frac{(\lambda-1)^{N-2}}{\lambda^N}&
-\frac{(\lambda-1)^{N-1}}{\lambda^{N+1}} 
  \\
0&   \frac{1}{\lambda}&-\frac{1}{\lambda^2}&\cdots & -\frac{(\lambda-1)^{N-3}}{\lambda^{N-1}}&
  -\frac{(\lambda-1)^{N-2}}{\lambda^{N }} 
  \\
 0& \ddots & \ddots & \ddots &\ddots &\vdots \\
  \vdots& \ddots & \ddots & \ddots &\ddots &\vdots \\
  \vdots& \ddots & \ddots & \ddots &\frac{1}{\lambda} &-\frac{1}{\lambda^2} \\
 0& \cdots & \cdots & \cdots & 0 &\frac{1}{\lambda} 
\\
\end{bmatrix},
$$

$$
(\Gamma_{\theta}+\widetilde{\Gamma})^{-1}=
(\Gamma_{\theta}-\widetilde{\Gamma})^{-T}-G_1\cdot\frac{
(\Gamma_{\theta}-\widetilde{\Gamma})^{-T}\bm{e} \bm{e}^T(\Gamma_{\theta}-\widetilde{\Gamma})^{-T}}
{1+G_1\cdot\bm{e}^T(\Gamma_{\theta}-\widetilde{\Gamma})^{-T}\bm{e}},
$$
 where $\lambda=\frac{2\theta}{G_1}+\frac{1}{2}$.
 \end{lem}
\begin{proof}[Proof of Lemma \ref{lem_inv_gamma}]
The computation of $(\Gamma_{\theta}-\widetilde{\Gamma})^{-1}$ is 
straightforward, see also proof of Proposition 3 of \cite{schied2018market}.
For $(\Gamma_{\theta}+\widetilde{\Gamma})^{-1}$
we may use the Sherman–Morrison formula.
Indeed, $(\Gamma_{\theta}+\widetilde{\Gamma})=
(\Gamma_{\theta}-\widetilde{\Gamma})^T+G_1\bm{e}\bm{e}^T
$, then since $(\Gamma_{\theta}-\widetilde{\Gamma})^T$
is non singular by the Sherman–Morrison formula we have the results,
$(\Gamma_{\theta}+\widetilde{\Gamma})^{-1}=
(\Gamma_{\theta}-\widetilde{\Gamma})^{-T}-G_1\cdot\frac{
(\Gamma_{\theta}-\widetilde{\Gamma})^{-T}\bm{e} \bm{e}^T(\Gamma_{\theta}-\widetilde{\Gamma})^{-T}}
{1+G_1\cdot\bm{e}^T(\Gamma_{\theta}-\widetilde{\Gamma})^{-T}\bm{e}}.$
\end{proof} 

 \begin{proof}[Proof of Proposition \ref{direc_solut_symm}]
    Without loss of generality we may assume that $G_1=1$. Indeed, let $\bm{\xi}_{1,\cdot}$ and $\bm{\xi}_{2,\cdot}$ be the admissible strategy for $X_1$ and $X_2$ respectively. Then, since $G_1$ is constant, we may scale the trading strategies by $G_1$. In particular, we may introduce 
    $\bm{\eta}_{i,\cdot}=G_1\cdot \bm{\xi}_{i,\cdot}$ for $i=1,2$, i.e.,
    the corresponding admissible strategies for the transformed inventory $Y_i=G_1\cdot X_i $,  $i=1,2$. Then, the two games are equivalent since,
   $  S_t^{\Xi}= S_t^0 -\sum_{t_k <t} G_1 (\xi_{1,k}+\xi_{2,k})=S_t^0 -\sum_{t_k <t} (\eta_{1,k}+\eta_{2,k}), 
   \ \forall\ t \in \numberset{T}$. 
   
  If we denote $(\Gamma_{\theta}-\widetilde{\Gamma})=A$,
   then by Lemma \ref{lem_inv_gamma},
   $(\Gamma_{\theta}+\widetilde{\Gamma})^{-1}=A^{-T}-\frac{
A^{-T}\bm{e} \bm{e}^TA^{-T}}
{1+\bm{e}^TA^{-T}\bm{e}}.$
However, 
$$(\Gamma_{\theta}+\widetilde{\Gamma})^{-1}\bm{e}=A^{-T}\bm{e}-A^{-T}\bm{e}\cdot\frac{
 \bm{e}^TA^{-T}\bm{e}}
{1+\bm{e}^TA^{-T}\bm{e}}=A^{-T}\bm{e}\cdot \bigg(1-\frac{
 \bm{e}^TA^{-T}\bm{e}}
{1+\bm{e}^TA^{-T}\bm{e}}\bigg)= \frac{
 A^{-T}\bm{e}}
{1+\bm{e}^TA^{-T}\bm{e}}.$$
   Thus, $\bm{e}^T(\Gamma_{\theta}+\widetilde{\Gamma})^{-1}\bm{e}=\frac{
 \bm{e}^T A^{-T}\bm{e}}
{1+\bm{e}^TA^{-T}\bm{e}}$, then
  \[
  \begin{split}
  \bm{v}&=\frac{1}{\bm{e}^T (\Gamma_{\theta}+\widetilde{\Gamma})^{-1}\bm{e}}(\Gamma_{\theta}+\widetilde{\Gamma})^{-1}\bm{e}=
  \frac{
 A^{-T}\bm{e}}
{\bm{e}^TA^{-T}\bm{e}}
  \\
     \bm{w}&=\frac{1}{\bm{e}^T (\Gamma_{\theta}-\widetilde{\Gamma})^{-1}\bm{e}}(\Gamma_{\theta}-\widetilde{\Gamma})^{-1}\bm{e}=
     \frac{
  A^{-1}\bm{e}}
{\bm{e}^TA^{-1}\bm{e}}.
        \end{split}\]
Moreover, if we denote $A^{-1}\bm{e}=\bm{x}$ and using the explicit formula for $A^{-1}$ we obtain that
$x_{N+1}=\frac{1}{\lambda}$, $x_n=\frac{1}{\lambda}-
\frac{1}{\lambda^2}\sum_{k=0}^{N-n} \left(\frac{\lambda-1}{\lambda}\right)^k=
\frac{1}{\lambda}-
\frac{1}{\lambda}(1-a^{N-n+1})=\frac{a^{N-n+1}}{\lambda}
$, for $n=1,\ldots,N$, where $\lambda=2\theta+\frac{1}{2}$ and $a=\frac{\lambda-1}{\lambda}$. Then, 
$\sum_{n=1}^{N+1}x_n=
\frac{a^N}{\lambda}
\sum_{n=1}^{N+1}a^{-(n-1)}=\frac{a^N}{\lambda}
\sum_{k=0}^{N}a^{-k}=
\frac{a^{N}}{\lambda}\cdot\frac{1-a^{-(N+1)}}{1-a^{-1}}=1-a^{N+1}$, since $1-a=\frac{1}{\lambda}$.
Therefore, $w_{N+1}=\frac{1}{\lambda \cdot (1-a^{N+1})}$ and 
$w_n=\frac{a^{N-n+1}}{\lambda \cdot (1-a^{N+1})}$ for $n=1,2,\ldots,N$. On the other hand, if 
    $A^{-T}\bm{e}=\bm{y}$,
$y_{1}=\frac{1}{\lambda}$, $y_n=\frac{1}{\lambda}-
\frac{1}{\lambda^2}\sum_{k=0}^{n-2} \left(\frac{\lambda-1}{\lambda}\right)^k=
\frac{1}{\lambda}-
\frac{1}{\lambda}(1-a^{n-1})=\frac{a^{n-1}}{\lambda}
$, for $n=2,\ldots,N+1$. Then, 
$\sum_{n=1}^{N+1}y_n=
\frac{1}{\lambda}
\sum_{n=1}^{N+1}a^{n-1}=\frac{1-a^{N+1}}{(1-a)\lambda}=1-a^{N+1}$.
Therefore, $v_1=\frac{1}{\lambda \cdot (1-a^{N+1})}$ and 
$v_n=\frac{a^{n-1}}{\lambda \cdot (1-a^{N+1})}$ for $n=2,\ldots,N+1$.
    \end{proof}

\begin{proof}[Proof of Theorem \ref{NE_symm}]
Without loss of generality, we may assume that the first agent is the Directional, so that $X_1=X_{direc}$ and $X_2=X_{arbi}=0$.
From \cite{schied2018market} the Nash equilibrium is 
provided by Eq. \eqref{eq_1_S&Z}-\eqref{eq_2_S&Z}, and we have Eq. 
\eqref{eq_1_S&Z_te}-\eqref{eq_2_S&Z_te}.
Then, since $\bm{v}$ and $\bm{w}$ are time-symmetric from Lemma \ref{direc_solut_symm} we have 
\[
\begin{split}
{v}_1\pm{w}_1 &=
{w}_{N+1}\pm{v}_{N+1}\\
{v}_k\pm{w}_k &=
{w}_{N+2-k}\pm{v}_{N+2-k}, \quad k=1,2,\ldots,N+1
\end{split}
\]
and it is straightforward to verify Eq. \eqref{eq_symm_NE_onedim}-\eqref{eq_symm_NE_2_onedim}
\end{proof}

\begin{proof}[Proof of Corollary \ref{cor_Ushape}]
W.l.o.g. we may assume $G_1=1$.
From Theorem \ref{NE_symm} we have the time-symmetry. 
On the other hand,
from the characterization of the fundamental solution
$\bm{v}$ and $\bm{w}$, see Proposition 
\ref{direc_solut_symm}, 
$\bm{\xi}_{direc,k}^*=\frac{X_{direc}}{2}\cdot \frac{a^{k-1}+a^{N-k+1}}{\lambda (1-a^{N+1})}$, 
where $\lambda=2\theta+1/2$ and $a=1-1/\lambda$.
Since $\theta >1/4$, then $\lambda >1$ and $a\in (0,1)$ so each component of the Nash equilibrium is positive. 
Let us denote with $\widetilde{\xi}_k$ the components of 
$\widetilde{\bm{\xi}}=\frac{2 \bm{\xi}_{direc,k}^*}{X_{direc}} $.
Let $\Delta_k=\widetilde{\xi}_{k+1}-\widetilde{\xi}_k$
be the first difference for $k=1,2,\ldots,N$. Then, since 
$\lambda (1-a^{N+1})>0$,
$\Delta_k=\frac{(a-1)(a^{k-1}-a^{N-k})}{\lambda (1-a^{N+1})}=
\frac{(a-1) a^{-k} (a^{2k-1}-a^N)}{\lambda (1-a^{N+1})}<0$
for $k<\lfloor{(N+1)/2}\rfloor$. For the convexity we consider 
the first difference of $\Delta_k$, i.e., $\Delta_k-\Delta_{k-1}$
for $k=2,3,\ldots,N$. However,
$\Delta_k-\Delta_{k-1}=\frac{(a-1)(a^{k-1}-a^{N-k})-(a-1)(a^{k-2}-a^{N-k+1})}{\lambda (1-a^{N+1})}=\frac{(a-1)(a^{k-1}-a^{N-k}-a^{k-2}+a^{N-k+1})}{\lambda (1-a^{N+1})}=\frac{(a-1)(a-1)(a^{k-2}+a^{N-k})}{\lambda (1-a^{N+1})}>0$ and we conclude.
\end{proof}

\begin{proof}[Proof of Theorem\ref{te_alter_defin}]
If $\Xi_1\neq0$ the matrix $\mathcal{M}$ is non singular.
\end{proof}

 \begin{proof}[Proof of Theorem \ref{id_intrinsic_decay}]
  Let us first consider the case when $N+1$ is even.
  From the left-hand side, since $\Pi$ is symmetric and $\bm{\xi}^*_{direc}$
  is time symmetric, the first $(N+1)/2$ equations are equal to the last $(N+1)/2$, where the first is equal to the $N+1$-th, 
  the second is equal to the $N$-th and so on.
  This, means that $\mbox{rk}(H)\leq (N+1)/2$. However, since $\bm{\xi}_{direc}^*$ has a U-shape, 
   in particular it is not constant but strictly decreasing with positive elements, and $\bm{g}$ 
   is strictly decreasing, then the first $(N+1)/2$ equations are different from each other, 
  since for each equation there are mixed products which are not contained in the remaining equations. Thus, $\mbox{rk}(H)=(N+1)/2$.
  The case of $N+1$ odd is straightforward, since using the same reasoning of the even case, we obtain that the first $N/2$ equations are equal to the last $N/2$ where
  in addition, we have another equation at the $N/2+1$ coordinate.
  So, $\mbox{rk}(H)=N/2+1.$
 \end{proof}
 \begin{proof}[Proof of Theorem
 \ref{te_linear_beta_results}]
 
From the discussion at the beginning of Section \ref{sec_linear_TIM} it is sufficient to show that 
the components of $\bm{\xi}^*_{direc}$ satisfy the relations \eqref{cond_linear_decay}.

W.l.o.g. we may assume the impact function of the market impact game be equal to $G_1=1$.
 From the characterization of the fundamental solutions 
 $\bm{v}$ and $\bm{w}$, see Proposition \ref{direc_solut_symm},
 ${\xi}_{direc,k}^*=\frac{X_{direc}}{2}\cdot \frac{a^{k-1}+a^{N-k+1}}{\lambda (1-a^{N+1})}$, 
 where $\lambda=2\theta+\frac{1}{2}$ and $a=1-1/\lambda$.
 So, let $\xi_k$ be the components of $\bm{\xi}_{direc}^*$, then, 
 for $k\leq\lfloor{N/2+1}\rfloor$,
 \[
 \frac{(\xi_{k-1}-\xi_k)}{X_{direc}-2\sum_{i=1}^{k-1}\xi_i}=
\frac{a^{k-2}+a^{N-k+2}-a^{k-1}-a^{N-k+1} }{2\lambda (1-a^{N+1})}\cdot\frac{1}{1-\sum_{i=1}^{k-1} \frac{a^{i-1}+a^{N-i+1}}{\lambda (1-a^{N+1})}   }=
 \]
 \[
 =\frac{a^{k-2}+a^{N-k+2}-a^{k-1}-a^{N-k+1} }{2\lambda (1-a^{N+1})}\cdot\frac{\lambda (1-a^{N+1})}{\lambda (1-a^{N+1})-\sum_{i=1}^{k-1} (a^{i-1}+a^{N-i+1})   }=
 \]
 \[
 =\frac{1}{2}\cdot \frac{a^{k-2}+a^{N-k+2}-a^{k-1}-a^{N-k+1} }{\lambda (1-a^{N+1})-\sum_{i=1}^{k-1} (a^{i-1}+a^{N-i+1})   }.
 \]
 Since $\sum_{i=1}^{k-1}a^{i-1}=(1-a^{k-1})/(1-a)$ and 
 $a^{N+1}\cdot\sum_{i=1}^{k-1} a^{-i}=(a^{N-k+2}-a^{N+1})/(1-a)$,
 and $\lambda(1-a)=1$
  \[
  \begin{split}
 &\frac{1}{2}\cdot \frac{a^{k-2}+a^{N-k+2}-a^{k-1}-a^{N-k+1} }{\lambda (1-a^{N+1})-\sum_{i=1}^{k-1} (a^{i-1}+a^{N-i+1})   }\\ =&\frac{1-a}{2}\cdot \frac{a^{k-2}+a^{N-k+2}-a^{k-1}-a^{N-k+1} }{\lambda(1-a) (1-a^{N+1})-(1-a^{k-1}+a^{N-k+2}-a^{N+1})  }
 \\
 =&
 \frac{1-a}{2}\cdot \frac{a^{k-2}+a^{N-k+2}-a^{k-1}-a^{N-k+1} }{a^{k-1}-a^{N-k+2} }=
 \frac{1-a}{2}\cdot \frac{a^{k-2}(1-a)-a^{N-k+1}(1-a) }{a^{k-1}-a^{N-k+2} }\\
 =&\frac{(1-a)^2}{2}\cdot \frac{a^{k-2}-a^{N-k+1}}{a^{k-1}-a^{N-k+2} } 
 =\frac{a(1-a)^2}{2a}\cdot \frac{a^{k-2}-a^{N-k+1}}{a^{k-1}-a^{N-k+2} } 
 =\frac{(1-a)^2}{2a}\cdot \frac{a^{k-1}-a^{N-k+2}}{a^{k-1}-a^{N-k+2} } \\
 =&\frac{(1-a)^2}{2a}.
\end{split}
 \]
 Therefore, $\frac{(\xi_{k-1}-\xi_k)}{X_{direc}-2\sum_{i=1}^{k-1}\xi_i}=\frac{(1-a)^2}{2a}$ and so it is independent for each $k\leq\lfloor{N/2+1}\rfloor$ and we conclude.
 \end{proof}

 \begin{proof}[Proof of Proposition \ref{prop_xi_spikes}]
From Example \ref{es_theta0} if $\theta_1=G_1/4$ then the Nash equilibrium of the market impact game is given by $\bm{\xi}_{direc}^*=[X_{direc}/2,0,\cdots,0 , X_{direc}/2]^T$.
Vice-versa, using the characterization of the Nash equilibrium, 
see Proposition \ref{direc_solut_symm} and Theorem \ref{NE_symm}, 
$\xi_{direc,1}^*=\frac{X_0}{2} \cdot \frac{(1+a^N)}{\lambda(1-a^{N+1})}$, where $a=1-1/\lambda$ and 
$\lambda=\frac{2\theta_1}{G_1}+\frac{1}{2}$, where $\theta_1\geq G_1/4$ so that $a\in [0,1)$.  Thus, if $\xi_{direc,1}^*=X_0/2$, then 
$1+a^N=\frac{1-a^{N+1}}{1-a}\iff 1+a^N=1+a+a^2+\cdots +a^N\iff
0=a+a^2+\cdots +a^{N-1}$, so $a=0$, i.e., $\theta_1=G_1/4.$
\end{proof}

\section{A Different Solution to Equation \eqref{eq_defi_alter_intrinsic_decay_kernel}}
\label{app_altern_approach2}
 
We now follow the same argument of Section \ref{sec_alternative_app} with a particular difference.

One could use $\bm{\xi}_1$ and $\bm{\xi}_2$, the Nash equilibrium solution of Schied and Zhang 
(e.g., using a linear constant decay kernel)
and obtain the (equilibrium) price dynamics $\bm{S}$,  from Equation \eqref{eq_alternative_system}.
Then, we ask the following question.
What is the impact faced by the only directional trader
when the directional does not consider the presence of other agents? 
Precisely, what is the propagator function inferred by directional traders if they solve the optimal execution problem in the standard Almgren-Chriss framework? 
In this case, the market impact considered by the directional is the same as that of the MIG, i.e., it is constant, even if the directional agent does not take into account any other competitors.
This is a substantial difference with respect to the previous approach.
Thus, once derived the optimal execution for the directional, $\bm{\xi}_{AC}=\frac{X_0}{N+1}\bm{e}$, we wonder what the market impact inferred by the trader is. To answer this question, we solve the Equation \eqref{eq_defi_alter_intrinsic_decay_kernel}, where $\Xi=\bm{\xi}_{AC}$,
 so that we can obtain the intrinsic decay kernel associated with the single-agent characterized by the price dynamics of the market impact games.

We first consider the same setting of Section \ref{sec_alternative_app}, wherein the market impact game there are the directional and one arbitrageur and the decay kernel is set to $G_1(t)\equiv1$. 
Then, we  compute the cumulative drift  $\bm{S}$ generated by the interaction of the two agents. We solve system \eqref{eq_defi_alter_intrinsic_decay_kernel} where $\mathcal{M}$ is computed by considering only the drift generated by the directional derived by the same model by without considering the presence of the arbitrageur, i.e., an 
Almgren and Chriss optimal schedule $\bm{\xi}_{AC}$ and 
we report in Figure \ref{fig:alternativeSZ2} the implied transient impact function $G_{intr}^{(P)}(t)$. 
Moreover, we compute $G_{intr}^{(P)}(t)$ when we consider a market impact game with $1$ directional trader and $2$ arbitrageurs using the \cite{luo_schied} model, where the decay kernel is set to $G_1(t)\equiv1$. 

As for the previous approach, we observe that the two implied transient impact kernels seem to differ in terms of ``size impact", i.e., the absolute value of $G_{intr}^{(P)}(0)$ is greater for the Luo and Schied game than those obtained starting from the Schied and Zhang model, 
but they exhibit the same shape. In particular, they are both positive and qualitatively decreasing functions.

The following motivation can interpret the shape of the fitted decay kernel.
Since the cumulative drift $\bm{S}$ is generated by the interactions of the fundamentalist and arbitrageur(s), it converges (in absolute value) to the inventory of the fundamentalist (1 in our example). Precisely,
due to the symmetry of Nash equilibrium,
the aggregate order flow  decreases to zero since in the last trading times, the arbitrageur(s) and the directional have opposite orders which compensate each other, see Figure \ref{fig:alternativeSZ_2}.
This shape is obtained by setting a constant market impact. Therefore when we consider a constant order flow, like the one obtained by the Almgren and Chriss optimal schedule $\bm{\xi}_{AC}$, we derive a strictly decreasing market impact function with an appropriate scaling factor provided by
$G(t_0)=(N+1) \cdot \Xi_1$.
Moreover, since we set $G_1(t)$ to a constant, the aggregate volume traded is directly proportional to the price dynamics, see Figure \ref{fig:alternativeSZ_2} right panels.

\begin{figure}[!t]
    \centering
    \subfloat[][]
    {\includegraphics[width=1\textwidth]{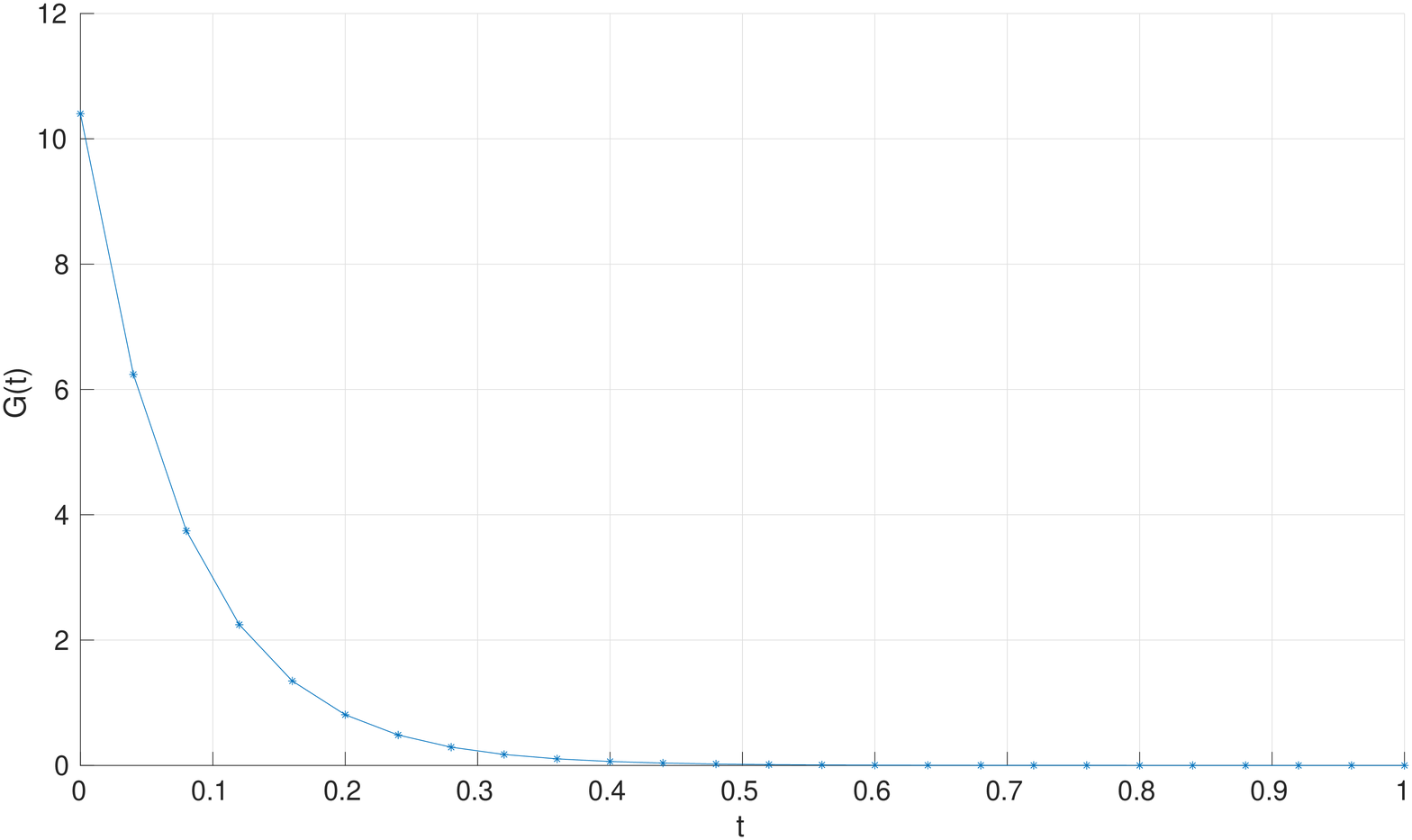}}\\
       \subfloat[][]{\includegraphics[width=1\textwidth]{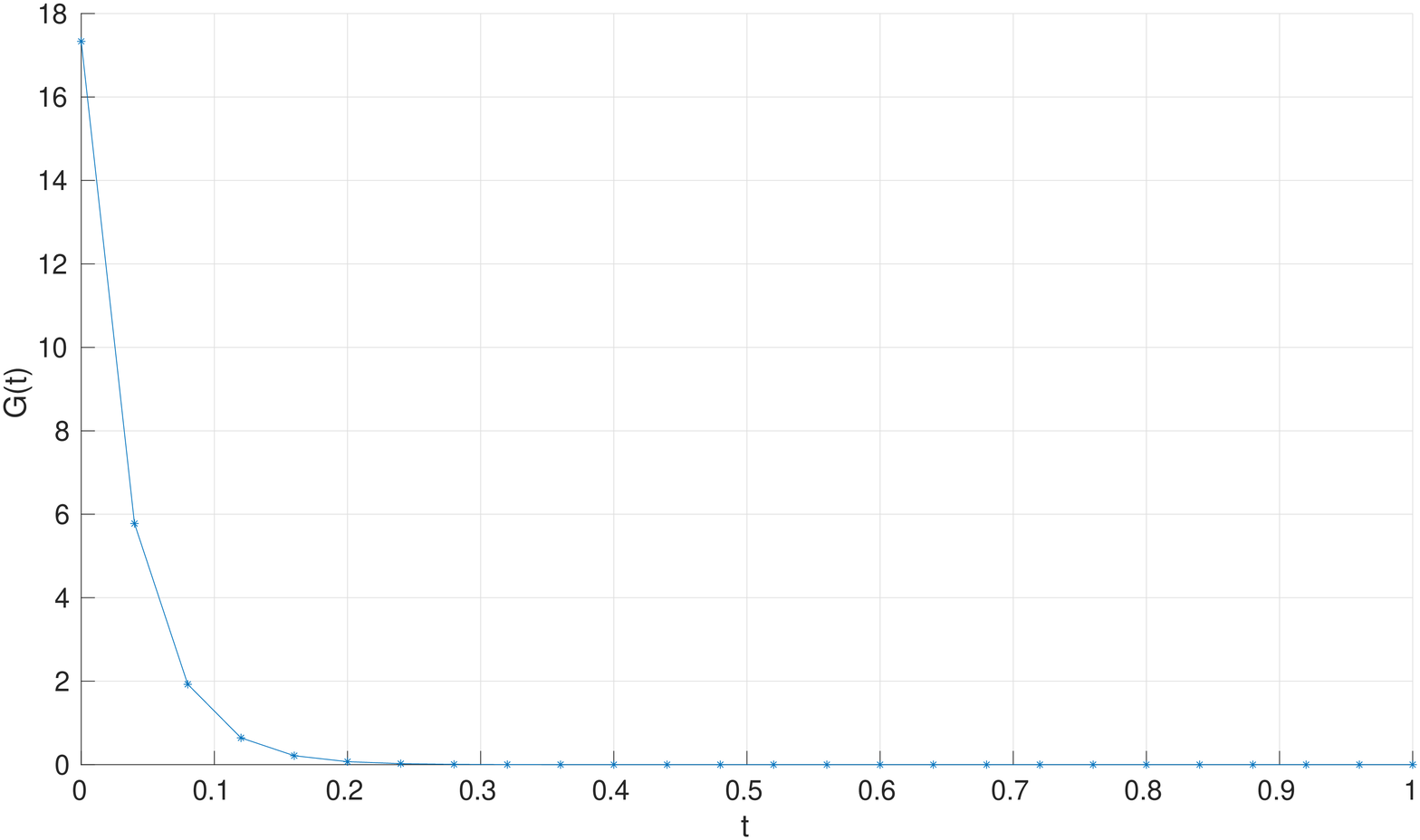}}
    \caption{The implied temporary impact function computed solving system \eqref{eq_defi_alter_intrinsic_decay_kernel} when we consider the cumulative drift generated by a Schied and Zhang game with a $1$ directional trader and $1$ arbitrageur, exhibit (a), and
    when we consider the cumulative drift generated by a Luo and Schied game with a $1$ directional trader and $2$ arbitrageur, exhibit (b).}
    \label{fig:alternativeSZ2}
\end{figure}

\begin{figure}[!t]
    \centering
    \subfloat[][]
    {\includegraphics[width=1\textwidth]{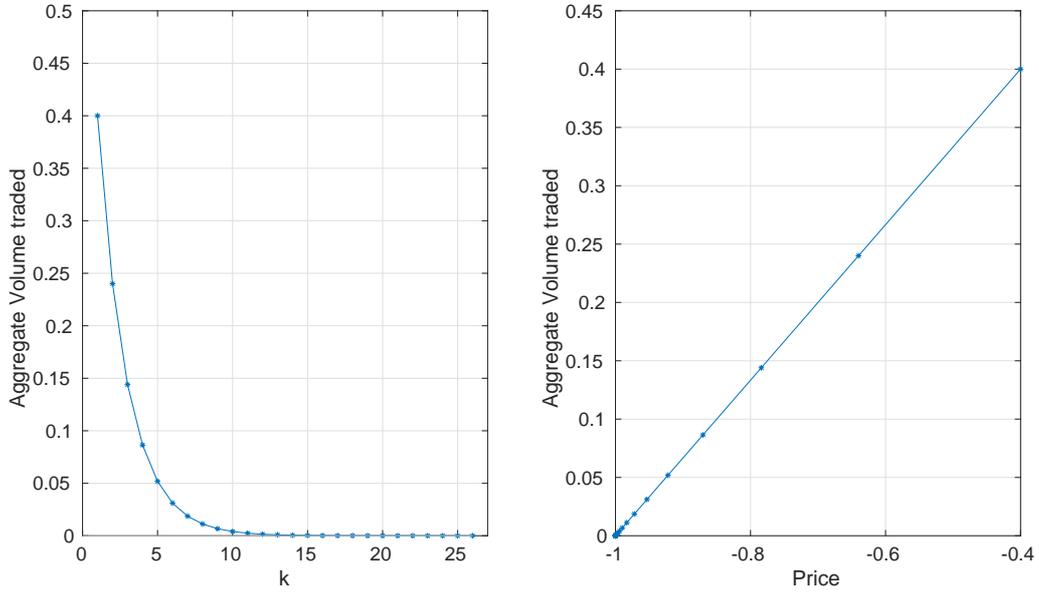}}\\
       \subfloat[][]{\includegraphics[width=1\textwidth]{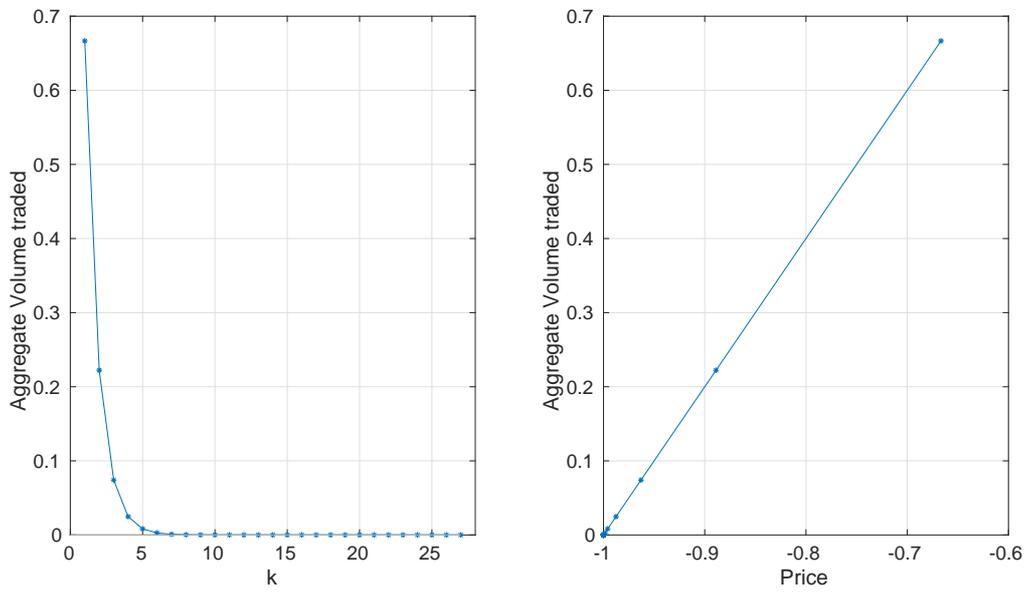}}
    \caption{Aggregate volume traded and its relation with price dynamics derived by a directional and an arbitrageur following the Schied and Zhang game and, panel (a), and by a directional and two arbitrageurs following a Luo and Schied game, panel (b).}
    \label{fig:alternativeSZ_2}
\end{figure}

\end{appendices}


\end{document}